\documentclass[10pt]{article}

\usepackage[letterpaper,hmargin=1in,vmargin=1in]{geometry}
\newcommand{\receiver}{\ensuremath{\mathcal{R}}}
\newcommand{\sender}{\ensuremath{\mathcal{S}}}
\newcommand{\opt}{\ensuremath{\operatorname{OPT}}}
\newcommand{\prob}{\ensuremath{\mathbb{P}}}

\newcommand{\rv}{\rho}
\newcommand{\receivervalue}{\rv}
\newcommand{\sv}{\xi}
\newcommand{\sendervalue}{\sv}
\newcommand{\vecrv}{\bm{\rv}}
\newcommand{\vecsv}{\bm{\sv}}
\newcommand{\HIRE}{\text{HIRE}}
\newcommand{\NOHIRE}{\text{NOT}}
\newcommand{\signal}{\sigma}
\newcommand{\stateON}{\theta}
\newcommand{\vecState}{\bm{\stateON}}
\newcommand{\vecstate}{\vecState}
\newcommand{\Ex}{\mathbb{E}}
\newcommand{\ap}{\ensuremath{\operatorname{SP}}}
\newcommand{\vecx}{\bm{x}}
\newcommand{\vecy}{\bm{y}}

\newcommand{\growingmid}{\mathrel{}\middle|\mathrel{}}
\newcommand{\D}{\displaystyle}

\usepackage{paralist}
\usepackage{tikz}
\usepackage{float}
\usepackage{amsmath}
\usepackage{amssymb}
\usepackage{amsthm}
\usepackage{graphicx}
\usepackage{wrapfig}
\usepackage[nocomma, short]{optidef}
\usepackage{bm}
\usepackage[ruled,vlined,linesnumbered]{algorithm2e}
\usepackage{nicefrac}
\usepackage[T1]{fontenc}
\usepackage{url}

\let\oldnl\nl
\newcommand{\nonl}{\renewcommand{\nl}{\let\nl\oldnl}}

\usepackage{multirow}

\newtheorem{claim}{Claim}

\newtheorem{theorem}{Theorem}[section]
\newtheorem{corollary}[theorem]{Corollary}
\newtheorem{lemma}[theorem]{Lemma}
\newtheorem{proposition}[theorem]{Proposition}

\renewcommand{\epsilon}{\varepsilon}

\allowdisplaybreaks

\title{The Secretary Recommendation Problem}

\author{Niklas Hahn\thanks{Institute of Computer Science, Goethe University Frankfurt/Main, Germany} \and Martin Hoefer\thanks{Institute of Computer Science, Goethe University Frankfurt/Main, Germany} \and Rann Smorodinsky\thanks{Faculty of Industrial Engineering \& Management, Technion, Israel}
}

\date{}

\begin{document}

\maketitle	

	\begin{abstract}
		In this paper we revisit the basic variant of the classical secretary problem. We propose a new approach in which we separate between an agent that evaluates the secretary performance and one that has to make the hiring decision. The evaluating agent (the sender) signals the quality of the candidate to the hiring agent (the receiver) who must make a decision. Whenever the two agents' interests are not fully aligned, this induces an information transmission (signaling) challenge for the sender. We study the sender's optimization problem subject to persuasiveness constraints of the receiver for several variants of the problem. 

        Our results quantify the loss in performance for the sender due to online arrival. We provide optimal and near-optimal persuasive mechanisms that recover at least a constant fraction of a natural utility benchmark for the sender. The separation of evaluation and decision making can have a substantial impact on the approximation results. While in some scenarios, techniques and results closely mirror the conditions in the standard secretary problem, we also explore conditions that lead to very different characteristics.
        
    \end{abstract}


	
	\section{Introduction}

	In the classical secretary problem a sequence of $n$ candidates with unknown valuations arrives sequentially in uniform random order. Upon arrival, a decision maker observes the value of a candidate and must make an irreversible decision of whether to hire the candidate or observe the subsequent one. The lion's share of this literature studies optimal or near-optimal algorithms for the decision maker and provides many elegant results on how well she can perform (see, e.g.,~\cite{Dynkin63,Ferguson89,Dinitz13, BabaioffIKK08,KesselheimRTV13,Kleinberg05,GoebelHKSV14,ChenHKLM15,FeldmanSZ18,HoeferK17,Reiffenhaeuser19,Rubinstein16} and many more).
	
	In this problem, the decision maker is tasked with two responsibilities: (1) evaluate the candidate, and (2) make the decision of whether or not to hire. In many settings, these two tasks are separated and the evaluation is done separately from the hiring decision. In this paper, we initiate the study of (near-)optimal persuasive mechanisms for the secretary problem when the evaluation is done by one agent while the decision is taken by another. Following tradition in the literature, we shall often refer to the latter as a `receiver' and to the former as a 'sender'. The separation between evaluation and hiring decision introduces incentives that are immaterial to the original setting.
	
	The aforementioned separation is wide-spread in applications. Consider, for example, the original motivating story of hiring an employee (e.g., a secretary). It is often the case that firms separate between the evaluation process, often led by a specialized HR department, and the hiring individual (the prospective boss). Another setting that has such a separation is in financial projects and investment opportunities. Many firms hire third-party consultants to evaluate a sequence of financial opportunities and use the consultant's report in their decision process. An opportunity not seized is often an opportunity foregone.
	
	To capture the incentive discrepancies between the sender and receiver, we endow each candidate with two valuations, one that is associated with the sender and one with the receiver. Candidates arrive according to a uniform random order. Upon arrival, the valuation pair is disclosed to the sender, who then makes some recommendation to the receiver. The receiver in turn takes one of two irrevocable actions, to hire or not to hire. Throughout the paper, we assume the sender has commitment power, and we design mechanisms for her recommendation such that the receiver always finds it in her interest to adopt the recommended action. Our goal is to design mechanisms that achieve (near-)optimal expected utility for the sender.
	
	Initially, we study an elementary scenario -- primarily for didactic reasons -- in which we assume that the valuation pairs of all candidates are known to both players in advance. Both players a-priori have full information about (1) the set of valuation pairs, (2) the fact that candidates arrive in uniform random order, and (3) the mechanism the sender uses for recommendation. In round $t$, the sender gets to see the arriving candidate and makes a recommendation using the mechanism. By applying Bayesian updates based on the a-priori information and the signals received in this and the previous rounds, the receiver can then decide whether it is in her interest to follow the recommendation or not. We provide optimal persuasive mechanisms for the sender, and we use this scenario as benchmark for the subsequent analysis. Hence, this is our \emph{basic} or \emph{benchmark scenario}. 

    Our main interest lies in the \emph{secretary scenario} as it is more reminiscent of the original secretary problem. Here the valuation pairs are unknown to the two players and are set by an adversary, but their arrival order is known to be uniform at random. The unknown adversarial design of valuations raises questions about modeling the incentives of the receiver. To determine her interests, the receiver should be able to evaluate what is ``in store'' for her if she adopts the sender's recommendation or deviates to another hiring strategy. To enable this decision, we adopt a robust approach and require that the mechanism is persuasive even if the receiver knows all valuation pairs (and, thus, has the same a-priori information as in the basic scenario above). Note that such a mechanism stays persuasive even if the receiver does not know the valuation pairs, since persuasiveness is merely a guarantee that it is in the receiver's interest to always follow the recommendation. Therefore, on a technical level, the secretary scenario differs from the basic scenario by the \emph{sender} a-priori lacking information about valuation pairs. In the secretary scenario, we provide a persuasive mechanism that guarantees a constant-factor approximation in the expected utility for the sender when compared to the optimum from the basic scenario. Hence, for a sender in a signaling context, \emph{secretary-style online arrival only results in a performance deterioration by a small constant factor}.

	Moreover, we study both scenarios in two different variants: In the first one, the receiver only receives the sender's signals; in the second one, she is also informed about the valuations of rejected candidates. We refer to this variant as a scenario \emph{with disclosure}. For motivation, one can think of candidates that are foregone as candidates that go on the market which, inter-alia, extracts and makes public all the relevant information about the candidate. In the example with financial investments, one can think of a sequence of investment opportunities given to some prominent investor (e.g., a top-ten private equity fund) before they go on the market. In the basic scenario with disclosure, we again obtain a mechanism with a constant approximation ratio. However, disclosure in the secretary scenario can lead to a deterioration in the expected utility by a factor of $\Theta(1/n)$.

	We examine all four scenarios (basic and secretary, both with and without disclosure) with different utilities that are standard in the literature on the secretary problem. With \emph{cardinal utility} a player strives to maximize her expected valuation of the hired candidate. In contrast, with \emph{ordinal utility} she strives to maximize the probability that her best candidate is hired. Our goal in all variants is to design good persuasive mechanisms for the sender. In almost all variants, we show constant approximation ratios, thus, \emph{even with online arrival and/or disclosure the expected utility for the sender deteriorates only by a small constant factor}. Also, we show that many ratios are optimal or asymptotically optimal with respect to any persuasive mechanism.



	\subsection{Related Literature}
	Our work can be seen as an extension of the celebrated \emph{secretary problem} which first appeared in print in Martin Gardner's 1960 \emph{Scientific American} column \cite{Gardner60} (but apparently originated much earlier, see \cite{Ferguson89}). The problem gained considerable popularity, and subsequently various extensions have been studied (c.f.\ \cite{Ferguson89} for an early survey and \cite{Dinitz13} for recent extensions). In all the literature we are familiar with, a separation between candidate evaluation and hiring decision is non-existent.
	
	In our model the sender has commitment power and hence our work contributes to the burgeoning literature on Bayesian persuasion, which originated in Aumann and Maschler \cite{AumannM66} and more recently enjoyed a renaissance through the work of Gentzkow and Kamenica \cite{KamenicaG11}. In particular, one can think of our model as an online Bayesian persuasion model, where the state of nature is revealed in a round-wise fashion to the sender who in turn sends a signal to the receiver in each round. Whereas dynamic Bayesian persuasion problems have been studied by various authors (e.g., \cite{Ely17, ElyFK15, Au15}), we are not familiar with an analysis of an online variant similar to what we study. Ely~\cite{Ely17} studies the case that the sender observes the current state and sends a message to the receiver. Using this information, the receiver's beliefs are updated, and she takes an action according to her current belief. The state of nature evolves according to a stochastic process, namely Poisson transitions over states. In principle, our approach might also be cast within a framework of a single, evolving state of nature, but the specific details of this approach would turn out to be very different from~\cite{Ely17}.
    Au~\cite{Au15} studies an approach, where in each round the receiver either takes an action or goes on to the next round. Unlike our model, the action is the same throughout. Once the action is taken, the process ends and both parties get their respective utilities. The sender extracts positive utility from the fact that the receiver has taken the action, regardless of the state of nature. Hence, the sender's objective is persuading the receiver to take the action. The receiver's utility depends on the state.
	A somewhat different approach is taken by Ely et al.\ \cite{ElyFK15} where the authors study the notions of suspense and surprise, i.e. the variance in the next update of belief and the distance of two consecutive beliefs, resp. The sender uses these objectives to design the disclosure policies.
	
	In the algorithmic literature, Dughmi and Xu \cite{DughmiX15} study the offline variant of Bayesian persuasion where candidates' valuations are drawn independently from known distributions. In~\cite{DughmiX17, BabichenkoB17}, the authors study algorithms and approximation for Bayesian persuasion with multiple receivers, each with a binary choice of actions. More recent additions to this literature include, e.g., duality considerations~\cite{DughmiNPW19} or the study of risk-conscious agents~\cite{AnunrojwongIL19}. In contrast, our approach extends the Bayesian persuasion model with a single sender and a single receiver in a different direction, by studying an online approach with adversarial valuations and random-order arrival.

	
	Somewhat related is previous work in the context of delegation~\cite{Holmstrom84}. Here, a principal delegates a search problem to an agent. Both have individual interests which might be misaligned. The principal can either accept or reject the agent's proposed solution and will do so according to a specified mechanism. A key difference between the two models lies in the commitment power. In our model, the party that observes the information designs the mechanism while in delegation, the mechanism is designed by the principal who only has a-priori information. Kleinberg and Kleinberg~\cite{KleinbergK18} recently studied how techniques and results from stochastic online optimization can be used to approximately solve delegation problems. In contrast to our work, they study application of these techniques to approximately solve an \emph{offline} delegation problem.
	
	\subsection{Model}
	We consider a Bayesian persuasion problem with online arrival. There is a \emph{sender} $\sender$ and a \emph{receiver} $\receiver$. A set of candidates arrives sequentially in $n$ rounds over time in \emph{uniform random order}. Each candidate comes with a pair of valuations, one for $\sender$ and one for $\receiver$. In the beginning, $\receiver$ knows the set of all these pairs and that their arrival order is uniform at random. $\sender$ always knows that the arrival order is uniform at random. In the \emph{basic} scenario, she also knows the set of valuation pairs, in the \emph{secretary} scenario she does not.

    In round $t$, a candidate arrives and a state of nature $\stateON_t \in \Theta_t$ is revealed to $\sender$. We use $\vecstate = (\stateON_1,\ldots,\stateON_n)$ to denote the vector of the states of nature revealed in all rounds $t \in \{1,\ldots,n\} = [n]$. Associated with $\stateON_t$ is the pair of non-negative values $\sendervalue(\stateON_t), \receivervalue(\stateON_t) \ge 0$, where $\rv(\stateON_t)$  is the utility of $\receiver$ and $\sv(\stateON_t)$ is the utility of $\sender$ when the candidate gets hired. We use $\vecrv$ and $\vecsv$ to denote the vectors of all possible utility values of all candidates in all states of nature, and $\rv_{\max} = \max_{\stateON_t} \rv(\stateON_t)$ and $\sv_{\max} = \max_{\stateON_t} \sv(\stateON_t)$ as well as $c_\sender = \arg \max_{\stateON_t} \sv(\stateON_t)$ and $c_\receiver = \arg \max_{\stateON_t} \rv(\stateON_t)$.\footnote{For simplicity, we assume w.l.o.g.\ that for all candidates all utility values $\rv_i$ and $\sv_i$ are mutually disjoint.} Having observed $\stateON_t$ and the corresponding values, $\sender$ transmits a signal $\signal_t$ to $\receiver$. Based on the signal, $\receiver$ then decides whether to hire the current candidate in round $t$ or not. Every decision is final and cannot be revoked later on. If $\receiver$ decides to hire, the process ends. Otherwise, round $t+1$ starts and the next candidate arrives. The process ends by the end of round $n$ at the latest. 

	Following~\cite{DughmiX15, KamenicaG11} we assume that there is commitment power, i.e., $\sender$ shall commit a-priori on a signaling strategy $\phi$, mapping each partially revealed vector of states of nature $(\stateON_1,\ldots,\stateON_t)$ to a signal $\signal_t$, for all $t \in [n]$. The order of events is (1) $\sender$ commits to a signaling strategy $\phi$, (2) first candidate arrives, (3) $\sender$ learns $\stateON_1$ and sends signal $\phi(\stateON_1)$ to $\receiver$, (4) $\receiver$ decides to hire or not, (5) repeat from step (2) with $\stateON_2$ and $\phi(\stateON_1,\stateON_2)$ etc., until a candidate gets hired or all $n$ candidates arrived. In a scenario \emph{without disclosure}, the sender's signals $(\phi(\stateON_1),\ldots,\phi(\stateON_1,\ldots,\stateON_t))$ are the only information for $\receiver$ in the beginning of round $t$ (in addition to the a-priori information). In a scenario \emph{with disclosure}, $\receiver$ also gets informed about the rejected candidates $(\stateON_1,\ldots,\stateON_{t-1})$ in the beginning of round $t$. Hence, with disclosure $\receiver$ knows the set of remaining candidates that will arrive in rounds $t,\ldots,n$.

    The goal of both $\sender$ and $\receiver$ is to maximize their individual objective. We consider two variants, for each $\sender$ and $\receiver$: The cardinal objective, where a player strives to maximize the \emph{expected utility}; and the ordinal objective, where the sender (receiver) wants to maximize the \emph{success probability}, i.e., the probability that the hired candidate is the optimal $c_\sender$ $(c_\receiver)$, resp. All probabilities/expectations are with respect to input randomization and internal randomization of the mechanism. Our mechanisms for $\sender$ are truthful-in-expectation mechanisms, where we strive to maximize the sender's objective.

    Applying a revelation-principle style argument~\cite{KamenicaG11,ArieliB16}, $\sender$ can restrict herself to a signaling strategy $\phi$ that is \emph{direct} and \emph{persuasive}. For direct signals, $\sender$ directly recommends one candidate for hire, i.e., $\signal_t \in \{\HIRE, \NOHIRE\}$. A direct mechanism $\phi$ is persuasive if $\receiver$ in every round maximizes her expected utility (from the candidate she hires eventually) by being obedient, i.e., always hiring upon receiving the signal $\HIRE$ and never hiring upon receiving the signal $\NOHIRE$. We assume that ties are broken in favor of $\sender$, i.e., only when $\receiver$ can strictly increase her success probability/expected utility, there is an incentive to deviate from the signals of $\phi$.

	\subsection{Results and Techniques}

    Our paper provides a comprehensive analysis of sixteen variants of the recommendation problem. For exposure, we here focus on two canonical cases in which $\sender$ and $\receiver$ either both have cardinal utilities (cardinal case) or both have ordinal utilities (ordinal case). Results for the remaining scenarios (e.g., cardinal sender and ordinal receiver) follow along the same lines. Throughout the paper, all asymptotics are in $n$, the number of candidates. 

    \paragraph{Cardinal Case} Our baseline is the \emph{basic scenario}, where we propose the Pareto mechanism (based on the Pareto procedure in Algorithm~\ref{algo:Pareto}), an optimal persuasive mechanism (Proposition \ref{prop:Pareto}) that determines an optimal signal based on the Pareto curve of the candidate set. We use the optimal sender utility (denoted by $\opt$) in this scenario as a benchmark for the other variants. $\opt$ can differ substantially from $\sv_{\max}$, can depend on the complete candidate set, and has no immediate closed-form expression. Perhaps interestingly, in the Pareto mechanism, if the sender-optimal candidate $c_\sender$ offers at least $\mu^{\receiver} = \frac{1}{n} \sum_i \rv_i$ to $\receiver$, then it gets recommended deterministically. Otherwise, the mechanism will compose a sender-optimal signal that has expected value $\mu^{\receiver}$ for $\receiver$, i.e., does not offer useful information.

    In the \emph{secretary scenario}, $\sender$ faces a problem with the main characteristics of the secretary problem, i.e., unknown candidate values and random-order arrival. The values are only revealed to $\sender$ once the candidate arrives. Observe that this scenario strictly generalizes the standard secretary problem. In particular, if $\sv_i = \rv_i$ for every candidate $i \in [n]$, the incentives of both players are perfectly aligned. In the basic scenario, when $\sender$ knows all candidates upfront, she would simply choose $c_\sender = c_\receiver$ and recommend it for hire. Thus, the benchmark in this case is expected utility $\opt = \sv_{\max}$. It is easy to see that in such instances no persuasive mechanism in the secretary scenario can beat the $1/e-o(1)$ guarantee from the classic optimal algorithm~\cite{Dynkin63}.

    When the incentives are not perfectly aligned, our mechanisms depart substantially from this classic template. Here, every persuasive mechanism must guarantee an expected utility for the receiver of $\mu^\receiver $. The Pareto mechanism optimizes the sender utility under this constraint. In the secretary scenario, we apply this technique adaptively to the set of arrived candidates. Perhaps surprisingly, this suffices to generate a distribution of $\HIRE$ signals, which leaves $\receiver$ with no additional information beyond the guarantee of expected utility $\mu^\receiver$. We obtain a persuasive mechanism and a $(1/(3\sqrt{3})-o(1))$-approximation of $\opt$ (Theorem \ref{theo:senderCardinalReceiverCardinalSecretaryNoDisclosure}).

    In the basic scenario with disclosure, we give a characterization of the optimal mechanism using an exponential family of nested linear programs. In terms of polynomial-time algorithms, we design another adaptive version of the Pareto mechanism w.r.t.\ the shrinking set of non-arrived candidates. It obtains an expected utility for $\sender$ of at least $1/3-o(1)$ times $\opt$ (Theorem \ref{theo:senderCardinalReceiverCardinalKnownUtilDisclosure}). 

    For the most challenging scenario with disclosure, we show that there are cases in which no persuasive mechanism can obtain an expected utility of more than $1/n$ -- while in the benchmark scenario the expected utility is $1/2$. Thus, in terms of approximation of $\opt$, no persuasive mechanism can have a better approximation ratio than $2/n$ (Theorem \ref{theo:senderOrdinalReceiverCardinalSecretaryDisclosure} and Corollary \ref{cor:senderCardinalReceiverCardinalSecretaryDisclosure}) -- while $1/n$ is always achieved by the trivial strategy of recommending a candidate uniformly at random. 
    
    \paragraph{Ordinal Case}
    Our results for the ordinal case are closer to existing results and techniques for the secretary problem. Our baseline is again the basic scenario. For ordinal utility, the receiver's only interest is to find her single best candidate. Without the sender signal this is quite hopeless, especially when $n$ grows large. Hence, it is not surprising that we can obtain a success probability of $1-o(1)$ for the sender. The benchmark is simply the standard benchmark of the secretary problem -- the best candidate $c_\sender$ for $\sender$. The intuition applies even in the basic scenario with disclosure -- we again obtain success probability of $1-o(1)$ for the sender (Proposition \ref{prop:receiverOrdinalKnownUtilDisclosure}).

    In the secretary scenario, we apply a variation of the classic template -- we mix the optimal algorithm for $\sender$ and for $\receiver$ and decide randomly upfront which version is applied. Although intuitive, it requires some effort to show that the mix retains persuasiveness for $\receiver$ even when it strongly favors $\sender$. The mechanism yields a success probability of $1/e-o(1)$ (Theorem~\ref{theo:receiverOrdinalSecretaryNoDisclosure}).

    In contrast, in the secretary scenario with disclosure, such an approach is not sufficient to incentivize $\receiver$ to follow the mechanism. The first round, in which a best-so-far candidate for $\receiver$ is rejected and revealed, $\receiver$ learns that the sender-optimal version is used. This can give $\receiver$ an incentive to ignore any upcoming $\HIRE$ signal. Instead, we run the two optimal algorithms in parallel and hire the first candidate that either algorithm would hire (first-come first-serve). We term this the First-Opt algorithm, and the success probability for $\sender$ becomes at least $1/4-o(1)$. Our main insight is that for negatively correlated utilities, this algorithm and the guarantee are indeed optimal, which we prove via a generalized dynamic programming technique (Theorems \ref{theo:receiverOrdinalSecretaryDisclosure} and \ref{thm:BeckmannUB}).

\begin{table}
	\centering
	\renewcommand{\arraystretch}{1.66}
	\begin{tabular}{c||c|c||c|c}
		& \multicolumn{2}{c||}{without Disclosure} & \multicolumn{2}{c}{with Disclosure} \\
		& Ordinal $\sender$ & Cardinal $\sender$ & Ordinal $\sender$ & Cardinal $\sender$\\  \hline \hline 
		Basic & \multicolumn{2}{c||}{Optimal mechanism {\footnotesize (Prop.\ \ref{prop:Pareto})}} & \textbf{1/2} {\footnotesize (Thm.\ \ref{theo:senderOrdinalReceiverCardinalKnownUtilDisclosureUB},  \ref{theo:senderOrdinalReceiverCardinalKnownUtilDisclosure})} & 1/3 {\footnotesize (Thm.\ \ref{theo:senderCardinalReceiverCardinalKnownUtilDisclosure})}  \\ \hline
		Secretary & 1/4 {\footnotesize (Thm.\ \ref{theo:senderOrdinalReceiverCardinalSecretaryNoDisclosure})} & $1/(3\sqrt{3})$ {\footnotesize (Thm.\ \ref{theo:senderCardinalReceiverCardinalSecretaryNoDisclosure})} & $\mathbf{\Theta(1/n)}$ {\footnotesize (Thm.\ \ref{theo:senderOrdinalReceiverCardinalSecretaryDisclosure})} & $\mathbf{\Theta(1/n)}$ {\footnotesize (Cor.\ \ref{cor:senderCardinalReceiverCardinalSecretaryDisclosure})}
	\end{tabular}
	\caption{Approximation guarantees of persuasive mechanisms for cardinal receiver utility discussed in this paper. All bounds stated without lower-order terms. Results indicated in bold have asymptotically matching upper bounds.}
	\label{table:lowerBoundsReceiverCard}
\end{table}	

	\begin{table}
	\centering
	\renewcommand{\arraystretch}{1.66}
	\begin{tabular}{c||c|c||c|c}
		& \multicolumn{2}{c||}{without Disclosure} & \multicolumn{2}{c}{with Disclosure} \\ 	
		& Ordinal $\sender$ & Cardinal $\sender$ & Ordinal $\sender$ & Cardinal $\sender$ \\  \hline \hline 
		Basic & \textbf{1} {\footnotesize (Prop.\ \ref{prop:receiverOrdinalKnownUtilNoDisclosure})} & \textbf{1} {\footnotesize (Prop.\ \ref{prop:receiverOrdinalKnownUtilNoDisclosure})} & \textbf{1} {\footnotesize (Prop.\ \ref{prop:receiverOrdinalKnownUtilDisclosure})} & \textbf{1} {\footnotesize (Prop.\ \ref{prop:receiverOrdinalKnownUtilDisclosure})}
			 \\ \hline 
		Secretary & \textbf{1/e} {\footnotesize (Thm.\ \ref{theo:receiverOrdinalSecretaryNoDisclosure})} & \textbf{1/e} {\footnotesize (Thm.\ \ref{theo:receiverOrdinalSecretaryNoDisclosure})} & \textbf{1/4} {\footnotesize (Thm.\ \ref{theo:receiverOrdinalSecretaryDisclosure}, \ref{thm:BeckmannUB})} & \textbf{1/4} {\footnotesize (Thm.\ \ref{theo:receiverOrdinalSecretaryDisclosure}, \ref{thm:BeckmannUB})}  \\
	\end{tabular}
	\caption{Approximation guarantees of persuasive mechanisms for ordinal receiver utility discussed in this paper. All bounds stated without lower-order terms. Results indicated in bold have asymptotically matching upper bounds.}
	\label{table:lowerBoundsReceiverOrd}
\end{table} 

    \subsection{Overview}
    In Table~\ref{table:lowerBoundsReceiverCard}, we summarize the approximation guarantees when $\receiver$ has cardinal utility, in Table~\ref{table:lowerBoundsReceiverOrd} the ones when $\receiver$ has ordinal utility. We view the basic scenario without disclosure as our benchmark scenario. Hence, all other entries in both tables represent approximation ratios with respect to the optimal value for $\sender$ that can be obtained in this case. In particular, for ordinal sender utility $\opt = 1-o(1)$, the optimal success probability in the benchmark scenario. Horizontal comparisons imply whether or not disclosure of rejected candidates to $\receiver$ comes at a cost for $\sender$. Vertical comparisons yield the loss in performance due to $\sender$ not knowing valuation pairs upfront. Observe that except for two variants (secretary with disclosure and cardinal receiver), all bounds are reasonable constants. Hence, when $\sender$ does not know the future or rejected candidates are revealed, it usually implies just a limited loss in her expected utility. 

    We complement the majority of our ratios with asymptotically tight upper bounds. Some of these upper bounds are new, such as $1/2$ in the basic scenario with disclosure, or $1/4$ and $O(1/n)$ in the secretary scenario with disclosure. These bounds imply structural differences between our signaling scenarios and the standard variant of the secretary problem. 

	The subsequent technical parts are organized as follows. Section~\ref{sec:cardinal} focuses on the case where the receiver maximizes her expected utility. The techniques in this section depart from the classical secretary problem. The section discusses the four scenarios of the problem: basic vs.\ secretary and disclosure vs.\ non-disclosure. Each scenario is analyzed for ordinal and cardinal sender utility. We open the section with the benchmark case (basic, non-disclosure) and conclude with the most complex one (secretary with disclosure). Each scenario is analyzed for cardinal and ordinal sender utility, respectively. For ordinal receiver utility, we proceed similarly in Appendix~\ref{sec:ordinal}.

	\section{Cardinal Utility for $\receiver$}
    \label{sec:cardinal}
	We consider cardinal receiver utility, i.e., $\receiver$ wants to maximize the expected utility of the hired candidate. To capture the sender's value (i.e., her success probability/expected utility), we use the Pareto frontier, a geometric interpretation of the candidate set. It consists of the upper boundary of the convex hull of the candidates' valuations. Hence, for increasing receiver-values, the Pareto frontier is non-increasing in the sender-values. This means that the value for $\sender$ of the Pareto curve at $\rv = \mu^{\receiver}$ is an obvious upper bound on the sender's expected utility.
	
	\begin{algorithm}[t]
		\caption{\label{algo:Pareto}Pareto Procedure}
		\DontPrintSemicolon
		\KwIn{A set of valuation pairs $(\rv_i, \sv_i)_{i \in [n]}$}
        Scale and normalize the candidate values, set $\mu^{\receiver} \leftarrow \sum_{i=1}^n\rv_i/n$ \;
		Let $\mathcal{P} = \{ (\rv_i,\sv_i) \mid  i \in \{1,\ldots,n\}\}$ and $conv(\mathcal{P})$ be the convex hull of $\mathcal{P}$\;
		Let $PC(\mathcal{P})$ be the Pareto frontier of $conv(\mathcal{P})$\;
        If $\rv_{c_\sender} \ge \mu^{\receiver}$, then set $a = b = c_\sender$; otherwise, find $a, b \in \{1,\ldots,n\}$ s.t.\ $(a,b)$ is the segment of $PC(\mathcal{P})$ that intersects with line $\rho = \mu^{\receiver}$ (see Fig.~\ref{fig:mechanism} for an illustration)  \hfil \texttt{// $a=b$ possible} \;
		Determine probability for candidate $a$:\\
		\nonl \hspace{0.1cm}If $\rv_a = \rv_b$ then $\alpha \leftarrow 1$; if $\rv_a \ne \rv_b$ and $\sv_a = \sv_b$, then $\alpha \leftarrow 0$; else  $\alpha \leftarrow \frac{\mu^{\receiver} - \rv_b}{\rv_a - \rv_b}$.\;
		Draw $x$ uniform from $[0,1]$. If $x \le \alpha$ then $c \leftarrow a$, else $c \leftarrow b$.\;
		\Return candidate $c$\;
	\end{algorithm}

	\subsection{Benchmark: Basic Scenario without Disclosure}

	In the benchmark scenario we assume that the point set $\mathcal{P} = \{ (\rv_i, \sv_i) \mid i \in [n] \}$ is known a-priori to both $\sender$ and $\receiver$. Only upon hiring, the identity of the hired candidate becomes known to $\receiver$. The identity of rejected candidates is not revealed. Consider the \emph{Pareto mechanism}. It is based on the Pareto procedure in Algorithm~\ref{algo:Pareto} that computes the best linear combination of points (denoted by \opt) within the convex hull of $\mathcal{P}$ that guarantees at least the average value $\mu_{\receiver}$ for $\receiver$. W.l.o.g.\ this is a linear combination of at most two points on the Pareto frontier. We interpret the linear combination as the probability distribution to signal \HIRE\ for the corresponding candidates (see Fig.~\ref{fig:mechanism} for an example). The procedure returns the chosen candidate $c$. The Pareto mechanism then signals $\sigma_t = \HIRE$ exactly in the round $t$ in which candidate $\theta_t = c$ arrives.

    We assume that internally the Pareto procedure works with a scaled and normalized version of its input. First, it scales all values in $\vecsv$ such that $\sv_{\max} = 1$, and all values in $\vecrv$ such that $\rv_{\max} = 1$. Second, in case of ordinal sender utility, it sets $\sv_{i'} = 0$ for all other candidates $i' \neq c_\sender$. In this way, the procedure can be applied to both instances with ordinal or cardinal sender utility, and its output is based on a lottery over candidates that maximizes success probability or expected utility for $\sender$, resp. Observe that in the scenarios discussed below, the procedure is applied with different subsets of candidate values as input. In this case, the procedure applies the internal adjustment w.r.t.\ the subset of candidates provided in the input.

	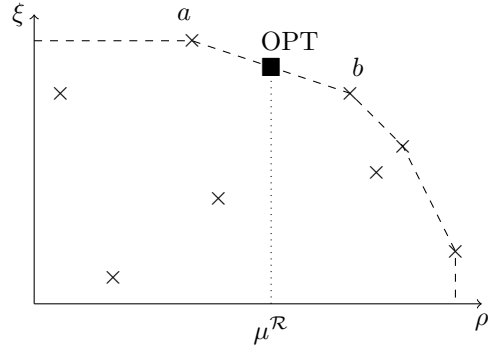
\begin{figure}[t]
		\centering
		\begin{footnotesize}
			\begin{tabular}{c||c|c|c|c|c|c|c|c||c|c}
				& 1 & 2 &  3 & 4 &  5 & 6  &  7 &  8 & $\mu$ & $\opt$\\ \hline \hline
				$\rv$ & 1 & 3 & 6  & 7 & 12 & 13 & 14 & 16 & 9   & 9\\ \hline
				$\sv$ & 8 & 1 & 10 & 4 &  8 &  5 &  6 &  2 & 6.5 & 9
			\end{tabular}
		\end{footnotesize}
		\hspace{0.25cm}
		\begin{tikzpicture}[baseline=2cm, scale = .35]
		\draw[->] (0,0) -- (17,0) node[below] {$\rv$}; 
		\draw[->] (0,0) -- (0,11) node[left] {$\sv$};

		\draw[dashed] (0,10) -- (6,10) -- (12,8) -- (14,6) -- (16,2) -- (16,0);
		
		\node[] at (1,8) (1) {$\times$};
		\node[]	at (3,1) (2) {$\times$};
		\node[]	at (6,10) (3) {$\times$};
		\node[] at (5.7,11) (a) {$a$};
		\node[] at (7,4) (4) {$\times$};
		\node[]	at (12,8) (5) {$\times$};
		\node[] at (12.3,9) (b) {$b$};
		\node[]	at (13,5) (6) {$\times$};
		\node[]	at (14,6) (7) {$\times$};
		\node[]	at (16,2) (8) {$\times$};
		
		\node[] at (9, -1)  (mur) {$\mu^{\receiver}$};
		\node[] at (9,9) (opt) {$\blacksquare$};
		\node[] at (9.7,10) (opt) {$\opt$};
		
		\draw[dotted] (9,0) -- (9,8.5);
		\end{tikzpicture}
		\caption{\label{fig:mechanism} Geometric interpretation of the Pareto mechanism. The instance has 8 candidates with values as shown in the table. Candidates $a=3$ and $b=5$ compose the optimal pair. $\HIRE$ signals are sent with probability $(\alpha, 1-\alpha) = (\nicefrac 12, \nicefrac 12)$. The optimal expected sender utility is $\opt = 9$.}
	\end{figure}

	\begin{proposition}
        \label{prop:Pareto}
		The Pareto mechanism is an optimal persuasive mechanism.
	\end{proposition}
	\begin{proof}
		Consider the Pareto mechanism and the event that $\receiver$ gets a signal $\signal_t = \HIRE$ in round $t$. We denote by $\vecx$ the distribution over candidates in round $t$ conditioned on this event. Clearly $\vecx$ is the same distribution, no matter in which round $t$ $\receiver$ gets a signal to hire. Now denote by $\vecy$ the distribution over candidates in round $t'$ conditioned on a $\HIRE$ signal in round $t \neq t'$. Clearly, $\vecy$ is the same distribution, for any $t,t'$ with $t \neq t'$.
		
		For $\receiver$, the expected utility of following a \HIRE\ signal is $\vecx^T \vecrv \ge \mu^{\receiver}$ by construction of the mechanism. By simply hiring in round 1 deterministically, $\receiver$ gets a value of $\mu^{\receiver}$ and a \HIRE\ signal with probability $1/n$. Hence,  $(\vecx + (n-1) \vecy)^T \vecrv = n \mu^{\receiver}$. Consequently, the expected utility of $\receiver$ for hiring upon a \NOHIRE\ signal is $\vecy^T  \vecrv \le \mu^{\receiver}$. Thus, the mechanism is persuasive.

    	Now consider any persuasive mechanism $\phi$ used by $\sender$. Let $x_i$ denote the probability that (over random arrival and randomization in $\phi$) candidate $i \in [n]$ is the first one with signal $\HIRE$. We use $\vecx$ to denote the vector. 
		Truthfulness of $\phi$ implies the following constraints: (1) $\vecx^T \vecrv \ge \mu^{\receiver}$, since this expected utility can be achieved by $\receiver$ simply by hiring in any fixed round; (2) $\| \vecx \|_1 \le 1$, since the mechanism can be assumed to recommend at most one candidate; (2) $\| \vecx \|_1 \ge 1$, since all $\rv_i \ge 0$ at least one candidate must be recommended (otherwise, $\receiver$ will deviate from $\phi$ by hiring in the last round). Hence, the distribution $\vecx$ resulting from the optimal persuasive mechanism is a feasible solution to the following maximization problem for the expected sender utility
		\begin{equation}
		\label{eq:basicLP}
		\begin{array}{lrcll}
		\renewcommand{\arraystretch}{0.5}
		\mbox{Max.} & \multicolumn{2}{l}{\bm{x}^T \vecsv} \\
		\mbox{s.t.} & \vecx^T \vecrv & \ge & \mu^{\receiver} \\
		& \|\vecx\|_1      & = & 1\\
		& x_i                          & \ge & 0 &\hspace{0.5cm} \mbox{for all } i=1,\ldots,n
		\end{array}
		\end{equation}
		It is straightforward to see that the Pareto mechanism computes a distribution that represents an optimum solution to the above LP. Hence, it is an optimal persuasive mechanism.
	\end{proof}

	\subsection{Secretary Scenario without Disclosure}
	The secretary scenario differs from the previous one by $\sender$ a-priori lacking knowledge of the valuation pairs. To cope with this challenge, we apply the Pareto procedure adaptively. Our mechanism $\phi(s)$ first samples a number of rounds and signals $\signal_t = \NOHIRE$ for rounds $t = 1,\ldots, s$. Then, in each round $t = s+1,\ldots,n-1$ let $A_t$ be the set of all previously arrived candidates, including the one currently under consideration in round $t$. The mechanism invokes the Pareto procedure on the set $A_t$. It signals $\signal_t = \HIRE$ if and only if the candidate chosen by the Pareto procedure has arrived in the current round $t$.

	\begin{algorithm}[t]
		\caption{\label{algo:GrowingPareto} Growing Pareto Mechanism}
		\DontPrintSemicolon
		\KwIn{Number of candidates $n$, sample size $s$}
		$A_0 \leftarrow \emptyset$ \;
		\For{$t=1$ to $n-1$}{
        $A_t \leftarrow A_{t-1} \cup \{\theta_t\}$\;
        \lIf{$t \le s$}{Signal $\NOHIRE$}
		\Else{$c_t \leftarrow $ candidate chosen by Pareto procedure on the set $A_t$ \\
			\lIf{$c_t = \stateON_t$}{Signal $\HIRE$ and end mechanism; \textbf{else} Signal $\NOHIRE$}}}
		Signal $\HIRE$ on $n$-th candidate.
	\end{algorithm}
	
	When $\signal_t = \HIRE$ and $\receiver$ deviates by refusing to hire, the mechanism signals \NOHIRE\ in every subsequent round. If $t=n$ and the mechanism has not signaled \HIRE\ so far, it sets $\signal_n = \HIRE$ deterministically in the last round. Due to the dependence on a growing candidate set $A_t$, we term this mechanism the \emph{Growing Pareto mechanism} (see Algorithm~\ref{algo:GrowingPareto}).

    The main results about the Growing Pareto mechanism are as follows. We first show that the mechanism is persuasive. We then prove the approximation guarantee of $1/4 - o(1)$ for the ordinal sender utility and $1/(3\sqrt{3}) - o(1)$ for the cardinal sender utility.

    Let us first consider persuasiveness. Consider round $t$. We follow~\cite{KesselheimRTV13} and rephrase the generation of the candidate arriving in round $t$ as follows: First choose the subset $A_t$ of candidates that arrived in rounds $1,\ldots,t$ uniformly at random from $[n]$, then pick the candidate arriving in round $t$ uniformly at random from $A_t$.

	\begin{lemma}
		\label{lem:signalT}
		Consider a given round $t$ and a given subset of candidates $A_t$ that arrived up to round $t$. In the Growing Pareto mechanism 
		$$ \Pr\left[\bigwedge_{i=1}^{t-1} \signal_i = \NOHIRE \growingmid A_{t-1} \right] = \begin{cases} 1 & t = 2,\ldots,s+1\\
		\frac{s}{t-1} & t = s+2 ,\ldots,n\end{cases}$$
		and
		$$\Pr[\signal_t = \HIRE \mid A_t] = \begin{cases} 0 & t = 1,\ldots,s\\
		\frac 1t \cdot \frac{s-1}{t-1} & t = s+1,\ldots,n-1\\
		\frac{s}{n-1} & t = n\end{cases} .$$
	\end{lemma}
	\begin{proof}
		Given that the set $A_t$ of candidates arrived up to round $t$, we draw all possible arrival sequences of $A_t$ in the first $t$ rounds in a reverse fashion. For $A_t$, the Pareto procedure singles out the candidates $a$ and $b$. The candidate $i_t \in [n]$ that arrives in round $t$ is chosen uniformly at random from $A_t$. Thus, the probability is $\frac 1t$ that the candidate is in $\{a,b\}$ and that in addition the Pareto procedure applied in round $t$ would return this candidate. For a signal $\signal_t = \HIRE$ of the Growing Pareto mechanism, however, it also needs to hold that $\signal_i = \NOHIRE$ for all $i \le t-1$.
		
		Given candidate set $A_t$ arrives in the first $t$ rounds and candidate $i_t$ in round $t$, consider the signal in round $t-1$. Now given the set $A_{t-1} = A_t \setminus\{i_t\}$, the Pareto procedure singles out some candidates $a$ and $b$. Since the candidate $i_{t-1}$ arriving in round $t-1$ is chosen uniformly at random from $A_{t-1}$, the probability that the Pareto procedure in round $t-1$ results in signal \NOHIRE\ is $\frac{t-2}{t-1}$.
		
		Suppose in round $t-1$ a candidate $i_{t-1}$ arrives and the Pareto procedure results in signal $\signal_{t-1} = \NOHIRE$. Then, we can apply the same argument for round $t-2$ and candidate set $A_{t-2} = A_t \setminus \{i_t, i_{t-1}\}$. Since we need a \NOHIRE\ signal in all previous rounds $i \in [t-1]$, we continue to apply the argument pointwise for all subsets $A_{i}$. Note that for $i \le s$ the probability $\Pr[\signal_i = \NOHIRE] = 1$ always.
		
		Hence, for every set $A_{t-1}$ we obtain 
		$$\Pr\left[\bigwedge_{i=1}^{t-1} \signal_i = \NOHIRE \growingmid A_{t-1} \right]= \frac{t-2}{t-1} \cdot \frac{t-3}{t-2} \cdots \frac{s}{s+1} = \frac{s}{t-1}
		$$
		for rounds $t = s + 2,\ldots,n$. Thus, for rounds $t = s+1,\ldots,n-1$ we see
		$$ \Pr[\signal_t = \HIRE \mid A_t] = \Ex_{i_t} \left[\Pr[\sigma_t = \HIRE \mid A_t, i_t] \cdot \Pr\left[\bigwedge_{i=1}^{t-1} \signal_i = \NOHIRE \growingmid A_t \setminus \{i_t\}\right]\right] = \frac{1}{t} \cdot \frac{s-1}{t-1}
		$$
		and, analogously, $ \Pr[\signal_t = \HIRE \mid A_t] = 1\cdot \frac{s}{n-1}$ for $t = n$.
	\end{proof}
	
	\begin{lemma}
        \label{lem:GrowingParetoIC}
		The Growing Pareto mechanism is persuasive in the secretary scenario without disclosure.
	\end{lemma}

	\begin{proof}
		Suppose the Growing Pareto mechanism sets $\signal_t = \HIRE$ in some round $t \in \{s+1,\ldots,n\}$ and $\sigma_{t'} = \NOHIRE$ for all $t' \in [t-1]$. Note $A_t$ is a subset chosen uniformly at random. For every $i \in A_t$, the probability to receive only signals \NOHIRE\ for $A_t \setminus \{i\}$ in rounds $1,\ldots,t-1$ is the same. Hence, by following the \HIRE\ signal, $\receiver$ obtains an expected utility of 
		$$ \Ex[\rv_t \mid \signal_t = \HIRE] \quad \ge \quad \Ex_{A_t}\left[\sum_{i \in A_t} \frac{\rv_i}{t}\right] \quad = \quad \mu^{\receiver}.$$
		Consider the event that $\receiver$ sees only $\NOHIRE$ in all rounds up to $t-1$. By Lemma~\ref{lem:signalT}, this event has the same probability for all subsets $A_{t-1}$. Hence, even conditioned on this event, the candidate arriving in round $t$ is uniformly distributed. Thus,
		$ \Ex\left[\rv_t \growingmid \bigwedge_{t'=1}^{t-1} \signal_{t'} = \NOHIRE\right] = \mu^{\receiver}$. Therefore, the expected utility of hiring when \NOHIRE\ is received also in round $t$ is $\Ex\left[\rv_t \growingmid \bigwedge_{t'=1}^{t} \signal_{t'} = \NOHIRE\right] \le \mu^{\receiver}$.
		
		Now consider the event that $\receiver$ observes $\sigma_t = \HIRE$. For all possible subsets $A_t$ this event has the same probability by Lemma~\ref{lem:signalT}. Hence, for any round $r > t$, candidate $\theta_r$ is uniformly distributed and gives expected value of $\mu^{\receiver}$. Thus, no $r > t$ yields a profitable deviation. The lemma follows.
	\end{proof}

	We now prove the approximation result, first for the case when $\sender$ has cardinal utility and aims to maximize her expected utility. The Growing Pareto mechanism provides a constant-factor approximation to the optimum in the corresponding benchmark scenario.

    \begin{theorem}
	   \label{theo:senderCardinalReceiverCardinalSecretaryNoDisclosure}
		In the secretary scenario without disclosure, the Growing Pareto mechanism with $s= \lfloor n/\sqrt{3} \rfloor$ yields a $\left(\frac{1}{3\sqrt{3}} - o(1)\right)$-appro\-xi\-ma\-tion of the optimal expected utility in the corresponding benchmark scenario.
    \end{theorem}

	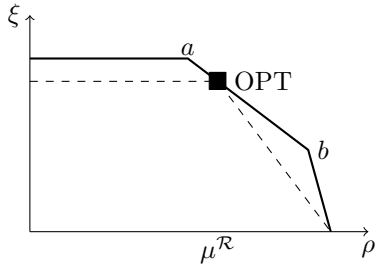
\begin{wrapfigure}{l}{0.45\textwidth}
		\centering
		\begin{tikzpicture}[xscale = .5, yscale=0.32]
		\draw[->] (0,0) -- (9,0) node[below] {$\rv$}; 
		\draw[->] (0,0) -- (0,9) node[left] {$\sv$};
		%
		\draw[thick] (0,7.2) -- (4.2,7.2) -- (7.4,3.4) -- (8,0) ;
		\draw[dashed]  (0,6.25) -- (5,6.25) -- (8,0);
		\node[] at (5, -0.75)   (mur) {$\mu^{\receiver}$};
		\node[] at (4.2,7.6) (a) {$a$};
		\node[] at (7.8,3.4) (b) {$b$};
		\node[] at (6.2,6.25) (opt) {$\opt$};
		\node[] at (5,6.25) (markeropt) {$\blacksquare$};
		\end{tikzpicture}
		\caption{\label{fig:reducedFrontier} Solid: Pareto frontier in the adapted instance in the proof of Lemma~\ref{lem:lossPerRoundCard}. Dashed: Lower bound on the Pareto frontier when $a$ and $b$ remain in the candidate set}
	\end{wrapfigure}
    The proof of this theorem relies on some auxiliary observations. Let $L = \{i \in [n] \mid \rv_i \le \mu^{\receiver}\}$ and $H = \{i \in [n] \mid \rv_i > \mu^{\receiver} \} = [n] \setminus L$ be a partition of the candidates into the ones with low receiver utility and high receiver utility, respectively. Let $d=|L|$ denote the cardinality of $L$. For the sake of analysis, we assume that $\rv_{\max} = 1$ and $\sv_{\max} = 1$. 

    We first concentrate on the case with cardinal sender utility. Consider the Pareto mechanism in the benchmark scenario. If the sender-optimal candidate has receiver utility $\rv_{c_\sender} \ge \mu^{\receiver}$, then $a=b$. The mechanism will wait for this particular candidate and recommend $\HIRE$ deterministically. Otherwise, $a \neq b$, and $\opt$ is both expected utility for $\sender$ and success probability for hiring $c_\sender$.

	For any subset of candidates $M \subseteq [n]$, we use $\opt_{-M}$ to denote the expected utility for $\sender$ in the benchmark scenario with candidate set $[n]\setminus M$. 
	Similar to $\opt_{-\{i\}}$, we define $\mu_{-i}^{\receiver} = \frac{1}{n-1}\left(\sum_{j=1}^n \rv_j - \rv_i\right)$ for $i \in [n]$. 

    \begin{lemma}\label{lem:lossPerRoundCard}
	   Consider an instance of the benchmark scenario with candidate set $[n]$. Let $a$ and $b$ be the candidates as determined in the Pareto procedure. It holds that $$\sum_{i \neq a,b} \opt_{-\{i\}} \ge (n-3) \opt .$$
    \end{lemma}

    \begin{proof}
	Consider the signaling scheme computed by the Pareto mechanism. If we adapt the instance and set $\sv_j = 0$ for all $j \neq a,b$, this does not change $\opt$, but it can only worsen the optima $\opt_{-\{i\}}$. Hence, in the new instance, the Pareto frontier of the convex hull of $\mathcal{P}$ consists only of the (at most) two points for candidates $a,b$, and a point for candidate $(\max_i \rv_i, 0)$ and the point $(0, \max_i \sv_i) = (0,\sv_a)$ (see Fig.~\ref{fig:reducedFrontier}).
	As described above, we partition the candidate set into the sets $L$ and $H$ and set $d = |L|$. The resulting value $\opt_{-\{i\}}$ is different depending on the candidate $i$ that gets removed from the pool. We underestimate the Pareto frontier with a line to $(\rv_{\max}, 0)$ (see the dashed line in Fig.~\ref{fig:reducedFrontier}). More formally, to estimate the loss in expected utility, use the slope $\frac{\opt}{\rv_{\max} - \mu^{\receiver}}$ from the point $(\mu^{\receiver}, \opt)$ representing the optimum to the point with maximal receiver value. The loss in sender utility caused by the change $\mu^{\receiver}_{-i} - \mu^{\receiver} = \sum_{j \in [n] \setminus \{i\}} \frac{\rv_j}{n-1} - \mu^{\receiver}$ can be bounded by

	\begin{align*}
	&\sum_{i \in L \setminus\{a\}}\opt - \opt_{-\{i\}} \quad \le \quad \frac{\opt}{\rv_{\max}-\mu^\receiver} \sum_{i \in L} \left[ \sum_{j \in [n] \setminus \{i\}} \frac{\rv_j}{n-1} - \mu^{\receiver}\right]\\
	&\quad = \quad \frac{1}{n-1}\cdot \frac{\opt}{\rv_{\max}-\mu^{\receiver}} \sum_{i \in L} \left[ \sum_{j \in L \setminus \{i\}} \rv_j + \sum_{j \in H} \rv_j - (n-1)\mu^{\receiver} \right]  \\
	&\quad = \quad \frac{\opt}{(n-1)(\rv_{\max}-\mu^{\receiver})} \sum_{j \in H} \left(\rv_j - \mu^{\receiver} \right) \\
    &\quad \le \quad \frac{\opt}{(n-1)(\rv_{\max}-\mu^{\receiver})} \sum_{j \in H} \left(\rv_{\max} - \mu^{\receiver} \right) \\
    &\quad = \quad \frac{n-d}{(n-1)} \opt \quad \le \quad \opt \enspace.
	\end{align*}
	If a candidate $i \in H$ arrives, $\mu^{\receiver}_{-i} \le \mu^{\receiver}$, and thus the expected payoff for $\sender$ can only increase. We can lower bound the payoff by $\opt$, c.f.\ Fig.~\ref{fig:reducedFrontier}, i.e. 
	$$\sum_{i \in H \setminus\{b\}} \opt - \opt_{-\{i\}} \le 0.$$
	Overall, this implies
	$$(n-2)\cdot \opt - \sum_{i \in [n] \setminus\{a,b\}} \opt_{-\{i\}} \le \opt$$
	which implies the lemma.
    \end{proof}

    Let $\opt_t$ be the expected value of the Pareto mechanism when applied to the benchmark scenario composed of the random subset $A_t$. Note that $\opt_n = \opt$, the optimum in the benchmark scenario.

    \begin{corollary}
	  \label{cor:valueRoundTCard}
	  For $t \ge 3$ it holds that $$\opt_t \ge \prod_{k=t+1}^{n} \left(1 - \frac{3}{k}\right) \opt = \frac{t(t-1)(t-2)}{n(n-1)(n-2)}\opt .$$
    \end{corollary}

    \begin{proof}
	We can generate the random set $A_t$ by starting with $[n]$ and iteratively removing a random candidate. Note that for $t=n-1$ we have by Lemma~\ref{lem:lossPerRoundCard}
	$$
	\opt_{n-1} = \frac{1}{n} \sum_{i \in [n]} \opt_{-\{i\}} \ge  \frac{1}{n} \sum_{i \in [n] \setminus \{a,b\}} \opt_{-\{i\}} \ge \frac{n-3}{n} \opt. $$
	The result for $t < n-1$ follows by repeated application.
    \end{proof}

    Finally, we prove Theorem~\ref{theo:senderCardinalReceiverCardinalSecretaryNoDisclosure}, the approximation result for cardinal sender utility.

    \begin{proof}[Proof of Theorem~\ref{theo:senderCardinalReceiverCardinalSecretaryNoDisclosure}]
	By combining the insights of Corollary~\ref{cor:valueRoundTCard} and Lemma~\ref{lem:signalT}, we see that in a given round $t = s+1,\ldots,n-1$ the Growing Pareto mechanism obtains an expected utility for $\sender$ of at least $\Pr[\signal_t = \HIRE] \cdot \opt_t$. For simplicity, we underestimate the utility in the last round $t=n$ by 0. For the expected utility of $\sender$, we use linearity of expectation over all rounds and set $s = \lfloor c\cdot n \rfloor$ for a constant $c$:
	\begin{align*}
	&\sum_{t=s+1}^{n-1} \frac{1}{t} \cdot \frac{s}{t-1} \cdot \frac{t(t-1)(t-2)}{n(n-1)(n-2)} \opt \quad
	= \quad \opt \cdot \; \frac{s}{n(n-1)(n-2)} \sum_{t=s+1}^{n-1} (t-2) \\
	&\quad = \quad \opt \cdot \; \frac{s}{n(n-1)(n-2)} \cdot \left( \frac{n(n-1)}{s} - \frac{s(s+1)}{2} - 2(n-1-s) \right) \\
	&\quad = \quad \opt \cdot \left(\frac{s}{n-2} - \frac{s^2(s+1)}{2n(n-1)(n-2)} - \frac{2s(n-1-s)}{n(n-1)(n-2)} \right) \\
	&\quad = \quad \opt \cdot \; (c - c^3 - o(1))\enspace.
	\end{align*}
    The last expression is maximized at $c = \frac{1}{\sqrt{3}}$, so we set $s = \lfloor n/\sqrt{3}\rfloor$. The theorem follows.
    \end{proof}

    We now proceed to the approximation result for ordinal sender utility.

	\begin{theorem}\label{theo:senderOrdinalReceiverCardinalSecretaryNoDisclosure}
		In the secretary scenario without disclosure, the Growing Pareto mechanism with $s= \lfloor n/2 \rfloor$ yields a success probability of at least $\left(\frac{1}{4}-o(1)\right)$ times the success probability in the corresponding benchmark scenario.
	\end{theorem}	

    We use the same notation as for the cardinal case above. In addition, let $\rv_{2nd}$ denote the second highest utility of any candidate for $\receiver$. For the Growing Pareto mechanism, the Pareto procedure requires cardinal values for all candidates in $A_t$. Since we assume ordinal sender utility, we set the sender utility to 1 for the best known candidate for $\sender$ in $A_t$, and 0 otherwise. Note that then $\opt$ becomes both -- expected utility for $\sender$ and success probability for hiring $c_\sender$. In step 1, the Pareto procedure then applies internal scaling and normalization to $A_t$ (as discussed above in Section~\ref{sec:cardinal}). It scales the receiver utilities from the input such that the best known candidate for $\receiver$ in $A_t$ is normalized to receiver utility 1.
    
    The next Lemma~\ref{lem:lossPerRoundOrd} represents an improved version of 	Lemma~\ref{lem:lossPerRoundCard}. Throughout the following analysis, we use $\opt_{-M}$ to denote the probability of hiring the best candidate for $\sender$ in $[n] \setminus M$. 
	\begin{lemma}\label{lem:lossPerRoundOrd}
		Let $\opt$ and $\opt_{-\{i\}}$ denote the expected utility in the benchmark scenario for candidate sets $[n]$ and $[n] \setminus \{i\}$, respectively.
		Then, the following holds:
		\[ \sum_{i \ne c_\sender,c_\receiver} \opt_{-\{i\}} + \opt \cdot \opt_{-\{c_\receiver\}} \ge \opt \left(n-2 - \frac{1}{n-1}\right)\enspace.
		\]
	\end{lemma}
\begin{proof}
	For the first case, we assume that the sender-optimal candidate $c_\sender$ has receiver utility $\rv_{c_\sender} \le \mu^{\receiver} \le \rv_{2nd}$. In this case, the Pareto procedure sets $a = c_\sender$ and $b = c_\receiver$, since $c_\sender$ is the only candidate with non-zero utility for $\sender$ and has receiver utility below $\mu^{\receiver}$. Since all other candidates have sender utility 0, the convex hull is composed of the segment between sender- and receiver-optimal candidates $c_\sender$ and $c_\receiver$.
	
	The resulting value $\opt_{-\{i\}}$ is different depending on the candidate that gets removed from the pool. If a candidate $i \in L$ with $\rv_i \le \mu^{\receiver}$ is removed, then $\mu^{\receiver}_{-i} \ge \mu^{\receiver}$, which implies that $\opt_{-\{i\}} \le \opt$. Note that upon removal of $b$, we have $\opt_{-\{b\}} = 1$ if $\mu_{-b}^{\receiver} \le \rv_a$. Otherwise, the new optimum point is located at $\mu_{-b}^{\receiver}$ and has value $\opt_{-\{b\}} = \frac{\rv_{2nd} - \mu_{-b}^{\receiver}}{\mu_{2nd} - \rv_a}$. Overall, we see that
	\begin{align*}
	&\sum_{i \ne c_\sender,c_\receiver} \opt_{-\{i\}} + \opt \cdot \opt_{-\{c_\receiver\}} \\
	&= \sum_{i \in L\setminus \{a\}} \opt_{-\{i\}} + \sum_{i \in H \setminus \{b\}} \opt_{-\{i\}} + \opt \cdot \min \left\{1, \frac{\rv_{2nd} - \mu_{-b}^{\receiver}}{\rv_{2nd}-\rv_a}\right\} \\
	&= \opt\left[(d-1) - \frac{1}{(n-1)(1-\mu^{\receiver})} \sum_{i\in L \setminus \{a\}} \left( \mu^{\receiver} - \rv_i \right) + (n-d-1) \right.\\*
	&\hspace{45pt} + \left. \sum_{i \in H \setminus \{b\}} \min \left\{\frac{\rv_i-\mu^{\receiver}}{(n-1)(1-\mu^{\receiver})},\frac{\mu^{\receiver}-\rv_a}{1-\mu^{\receiver}}\right\} + \frac{\rv_{2nd} - \mu_{-b}^{\receiver}}{\rv_{2nd}-\rv_a} \right] \\
	&\ge \opt\left[(n-2) - \frac{1}{(n-1)(1-\mu^{\receiver})} \sum_{i\in L} \left( \mu^{\receiver} - \rv_i \right)  \right.\\*
	&\hspace{45pt} + \left. \sum_{i \in H \setminus \{b\}} \min \left\{\frac{\rv_i-\mu^{\receiver}}{(n-1)(1-\mu^{\receiver})},\frac{\mu^{\receiver}-\rv_a}{1-\mu^{\receiver}}\right\} + \frac{\rv_{2nd} - \mu^{\receiver}}{\rv_{2nd}-\rv_a} \right] \\
	&= \opt\left[(n-2) + \frac{\rv_{2nd} - \mu^{\receiver}}{\rv_{2nd}-\rv_a} - \frac{1}{n-1}  \right.\\*
	&\hspace{45pt} - \left. \frac{1}{(n-1)(1-\mu^{\receiver})} \sum_{i\in H \setminus \{b\}} \left( \left(\rv_i - \mu^{\receiver}\right) - \min \left\{\rv_i-\mu^{\receiver},(\mu^{\receiver}-\rv_a)(n-1)\right\} \right)  \right] \enspace.\\
	\end{align*}
	
	We consider two subcases.	
	\begin{description}
		\item[Subcase $\mu^{\receiver} < \frac{\rv_{2nd} + (n-1) \rv_a}{n}$:] Here, we see
		\begin{align*}
		&\sum_{i \ne c_\sender,c_\receiver} \opt_{-\{i\}} + \opt \cdot \opt_{-\{c_\receiver\}}\\
		&\ge \opt\left[(n-2) + \frac{\rv_{2nd} - \mu^{\receiver}}{\rv_{2nd}-\rv_a} - \frac{1}{n-1}  \right.\\
		&\hspace{45pt} - \left. \frac{1}{(n-1)(1-\mu^{\receiver})} \sum_{i\in H \setminus \{b\}} \left( \left(\rv_i - \mu^{\receiver}\right) - \min \left\{\rv_i-\mu^{\receiver},(\mu^{\receiver}-\rv_a)(n-1)\right\} \right)  \right] \\
		&\ge  \opt\left[(n-2) + 1 - \frac{\mu^{\receiver}-\rv_a}{\rv_{2nd}-\rv_a} - \frac{1}{n-1} - \frac{n-2}{n-1} \right] \\
		&\ge \opt \left[n-2 - \frac{1}{n} \right]\enspace.
		\end{align*}
		For the penultimate line we used that $\rv_i \le 1$ and the $\min$ is at least 0. In the last inequality, we used the assumption that $\mu^{\receiver} < \frac{\rv_{2nd} + (n-1) \rv_a}{n}$ and thus 
		\[
		- \frac{\mu^{\receiver}-\rv_a}{\rv_{2nd} - \rv_a} \ge - \frac{\rv_{2nd}-\rv_a}{n(\rv_{2nd}-\rv_a)} = -\frac{1}{n}\enspace.
		\]
		
		\item[Subcase $\mu^{\receiver} \ge \frac{\rv_{2nd} + (n-1) \rv_a}{n}$:] In this case, we see \[\sum_{i\in H \setminus \{b\}} \left( \left(\rv_i - \mu^{\receiver}\right) - \min \left\{\rv_i-\mu^{\receiver},(\mu^{\receiver}-\rv_a)(n-1)\right\} \right)=0\] as $\rv_i \le \rv_{2nd}$ for all $i \in H \setminus \{b\}$ and $\rv_{2nd} - \mu^{\receiver} \le (\mu^{\receiver}-\rv_a)(n-1)$.
		Thus,
		\begin{align*}
		&\sum_{i \ne c_\sender,c_\receiver} \opt_{-\{i\}} + \opt \cdot \opt_{-\{c_\receiver\}} \\
		&\ge \opt\left[(n-2) + \frac{\rv_{2nd} - \mu^{\receiver}}{\rv_{2nd}-\rv_a} - \frac{1}{n-1}  \right.\\
		&\hspace{45pt} - \left. \frac{1}{(n-1)(1-\mu^{\receiver})} \sum_{i\in H \setminus \{b\}} \left( \left(\rv_i - \mu^{\receiver}\right) - \min \left\{\rv_i-\mu^{\receiver},(\mu^{\receiver}-\rv_a)(n-1)\right\} \right)  \right] \\
		&=  \opt\left[(n-2) + \frac{\rv_{2nd} - \mu^{\receiver}}{\rv_{2nd}-\rv_a} - \frac{1}{n-1} \right] \\
		&\ge \opt \left[n-2 - \frac{1}{n-1} \right]\enspace.
		\end{align*}
	\end{description}
	
	Now, consider the second case in which the sender-optimal candidate $c_\sender$ has receiver utility $\rv_{c_\sender}, \rv_{2nd} \le \mu^{\receiver}$. Again, in this case the Pareto procedure chooses $a = c_\sender$ and $b = c_\receiver$. We can bound as follows
	\begin{align*}
	&\sum_{i \ne c_\sender,c_\receiver} \opt_{-\{i\}} + \opt \cdot \opt_{-\{c_\receiver\}} = \sum_{i \in L \setminus \{a\}} \opt_{-\{i\}} + \opt \cdot \opt_{-\{c_\receiver\}} \\
	&\ge (n-2)\opt - \frac{1}{(n-1)(1-\mu^{\receiver})} \sum_{i \in L } \left( \mu^{\receiver} - \rv_i \right) \\
	&= \left(n-2-\frac{1}{n-1}\right)\opt\enspace.
	\end{align*}
	
	Finally, consider the third case, $c_\sender$ has $\rv_{c_\sender} > \mu^{\receiver}$. Then, the Pareto procedure chooses $a = b = c_\sender$, and $\opt_{-c_\receiver} = 1$. We can bound as follows
	\begin{align*}
	&\sum_{i \ne c_\sender,c_\receiver} \opt_{-\{i\}} + \opt = \sum_{i \in L} \opt_{-\{i\}} + \sum_{i \in H \setminus \{c_\sender, c_\receiver\}} \opt_{-\{i\}} +  \opt \\
	&\ge \opt \left(d - \frac{1}{(n-1)(1-\mu^{\receiver})} \sum_{i \in L} \left(\mu^{\receiver}-\rv_i\right) \right) + \sum_{i \in H \setminus \{c_\sender, c_\receiver\}} \opt + \opt  \\
	&\ge \opt \left(n-1 - \frac{1}{(n-1)(1-\mu^{\receiver})} (n-d) \left(1 - \mu^{\receiver}\right) \right) \\
	&\ge  (n-2)\opt \quad \ge \quad \left(n-2-\frac{1}{n-1}\right)\opt\enspace.
	\end{align*}
\end{proof}
	Let $\opt_t$ be the success probability of the Pareto mechanism when applied to the benchmark scenario with random subset $A_t$. We can repeatedly apply Lemma~\ref{lem:lossPerRoundOrd} in the same way as Lemma~\ref{lem:lossPerRoundCard} in Corollary~\ref{cor:valueRoundTCard} above. 
	\begin{corollary}\label{cor:valueRoundTOrd}
		For all $t \in [n]$ it holds that $\opt_t \ge \frac{(t-2)(t-1)}{(n-2)(n-1)}\opt . $
	\end{corollary}

	\begin{proof}
	The random subset $A_t$ can be generated by starting from the candidate set $[n]$ and iteratively removing a candidate picked uniformly at random. In a single step
	\begin{align*}\opt_{n-1} \quad &= \quad \frac{1}{n} \sum_{i \in [n]} \opt_{-\{i\}} \quad
	= \quad \frac{1}{n} \sum_{i \ne c_\sender,c_\receiver} \opt_{-\{i\}} + \frac{1}{n} \opt_{-\{c_\receiver\}} \\ 
	&\ge \quad \frac{1}{n} \sum_{i \ne c_\sender,c_\receiver} \opt_{-\{i\}} + \frac{1}{n} \opt \cdot \opt_{-\{c_\receiver\}} \quad
	\ge \quad \frac{n-2-\frac{1}{n-1}}{n} \cdot \opt
	\end{align*}
	due to $\opt_{-\{c_\sender\}} = 0$, $\opt \le 1$, and Lemma \ref{lem:lossPerRoundOrd}.
	Repeated application of this property implies 
	\begin{align*}
	\opt_t \quad &\ge \quad \opt \prod_{k=t+1}^n \frac{k-2-\frac{1}{k-1}}{k} \quad 
	\ge \quad \opt \prod_{k=t+1}^n \frac{k-3}{k-1} \\
	&= \quad \opt \cdot \frac{t-2}{t} \cdot \frac{t-1}{t+1} \cdot \frac{t}{t+2} \cdot \dots \cdot \frac{n-4}{n-2} \cdot \frac{n-3}{n-1} \quad
	= \quad \opt \cdot \frac{(t-2)(t-1)}{(n-2)(n-1)}\enspace.
	\end{align*}
	\end{proof}

    Now, using Lemma~\ref{lem:lossPerRoundOrd} and Corollary~\ref{cor:valueRoundTOrd}, Theorem~\ref{theo:senderOrdinalReceiverCardinalSecretaryNoDisclosure} can be proved similar as Theorem~\ref{theo:senderCardinalReceiverCardinalSecretaryNoDisclosure} above. 
    
    \begin{proof}[Proof of Theorem~\ref{theo:senderOrdinalReceiverCardinalSecretaryNoDisclosure}]
    	In a given round $t = s+1, \dots, n-1$, the Growing Pareto mechanism obtains a success probability of at least $\Pr[\sigma_t = \HIRE] \cdot \opt_t$. We underestimate the success probability in round $n$ by 0. Using linearity of expectation and choosing a sample size $s = \lfloor c\cdot n \rfloor$ with a constant $c$, we obtain
    	\begin{align*}
    		&\sum_{t= s+1}^{n-1} \frac{1}{t} \cdot \frac{s}{t-1} \cdot \frac{(t-1)(t-2)}{(n-1)(n-2)} \opt
    		\quad = \quad \opt \cdot \frac{s}{(n-1)(n-2)} \sum_{t=s+1}^{n-1} \frac{t-2}{t} \\
    		&\quad = \quad \opt \cdot \left(\frac{s}{n-2} - \frac{s^2}{(n-1)(n-2)} - \frac{2s \cdot (H_{n-1} - H_{s+1})}{(n-1)(n-2)}  \right) 
    		\quad = \quad \opt \cdot \left(c - c^2 - o(1)\right)\enspace.
    	\end{align*}
    	Here, $H_k = \sum_{i=1}^k 1/i$ is the $k$-th harmonic number. The last expression is maximized at $c=1/2$, therefore we choose $s = \lfloor n/2 \rfloor$.
    	The theorem follows.
    \end{proof}

	\subsection{Basic Scenario with Disclosure}
	In this section, we consider the basic scenario in which rejected candidates are revealed to $\receiver$.  We give a characterization of the optimal mechanism using backwards induction. We then show that there are instances in which no online mechanism can achieve more than a fraction of $1/2$ of $\opt$, i.e., the success probability/expected utility obtained in the benchmark scenario. We describe and analyze the polynomial-time persuasive Shrinking Pareto mechanism. It obtains a constant-factor approximation of $\opt$ with a ratio of $1/2-o(1)$ (hence, asymptotically optimal) for ordinal sender utility. For cardinal sender utility the ratio is $1/3-o(1)$.

    Let us start by deriving an optimal mechanism via backwards induction. Again, we utilize the results obtained in the case of cardinal sender utility to express bounds for the ordinal sender case. By setting $\sv_{\max}=1$ and $\sv_i =0$ for all sender-suboptimal candidates $i \ne c_\sender$, success probability and expected utility are the same value. Hence, the following result holds for both sender objectives. For simplicity, we state and prove the theorem for cardinal sender utility.
	
	An optimal mechanism for the scenario can be computed by backward induction and solution of a series of linear programs. These programs express the optimal decision given a subset $C \subseteq [n]$ of candidates and given the optimal decision policy computed for each subset $C \setminus \{i\}$, for every $i \in C$.  It is an interesting open problem whether there is a more intricate, polynomial-time algorithm to compute the optimal persuasive mechanism in this scenario.

	\begin{theorem}\label{thm:disclosureLPs}
		An optimal persuasive mechanism for the sender's expected utility in the basic scenario with disclosure can be computed by solving $2^n$ linear programs.
	\end{theorem}

\begin{proof}
	Suppose we reach the beginning of round $t$ with a set $C$ of $n-t+1$ remaining candidates. For every candidate $i \in C$, suppose the sender sees candidate $i$ in round $t$ and has computed the signaling policy for an optimal persuasive mechanism for rounds $t+1,\ldots,n$ and subset $C \setminus \{i\}$. 
	
	Clearly, in round $t = n$ the sender should set $\sigma_n = \HIRE$ with probability 1. This is the only persuasive signaling scheme. Now suppose we are in round $t < n$. We denote by 
	$$x_i^C = \Pr[\sigma_t = \HIRE \mid i\in C \text{ arrives in round } t].$$
	In case $i$ is not hired although $\sender$ signaled \HIRE, $\sender$ will never signal \HIRE\ again. Let $u^{\sender}_{C \setminus \{i\}}$ be the expected utility of $\sender$ from the optimal persuasive mechanism applied from round $t+1$ onwards with candidate set $C \setminus \{i\}$. Similarly, we define $u^{\receiver}_{C \setminus \{i\}}$ for $\receiver$. Assuming the signaling scheme $\phi$ determined by $x$ is persuasive, then the expected utility obtained by $\sender$ 
	\[ \frac{1}{|C|} \sum_{i \in C} (x_i^{C} \sv_i + (1-x_i^{C})\cdot u^{\sender}_{C \setminus  \{i\}}) \]
	is the objective function of a linear program. In this LP we have the obvious constraints $x_i^C \in [0,1]$, and two constraints on $x$ that ensure persuasiveness. Suppose $\receiver$ gets signal \HIRE. This happens with total probability $\Pr[\sigma_t = \HIRE \mid C] = \sum_{i \in C} x_i^C/|C|$. We assume w.l.o.g.\ that this probability is positive, otherwise there is nothing to prove. The probability that the recommended candidate is $i$ is $p^h_i = \frac{x_i^C}{\sum_{j \in C} x_j^C}$. Upon compliance with the \HIRE\ signal, $\receiver$ gets a utility of $\sum_{i \in C} p^h_i\rv_i$. Upon deviation, the candidate is rejected and $\sender$ stops to signal \HIRE. Thus, the expected utility of $\receiver$ becomes $\sum_{i \in C} p_i^h \cdot \frac{1}{|C|-1}\sum_{j \in C \setminus \{i\}} \rv_j$. Hence, persuasiveness requires
	\begin{equation}
	\label{eq:optRev}
	\begin{array}{rrcl}
	&  \D\sum_{i \in C} p^h_i \rv_i &\ge& \D \sum_{i \in C}p^h_i \cdot \frac{1}{|C|-1} \sum_{j \in C \setminus \{i\}} \rv_j\\
	\Leftrightarrow \; &  \D\sum_{i \in C} x_i^C \rv_i &\ge& \D \sum_{i \in C}x_i^C \cdot \frac{1}{|C|-1} \sum_{j \in C \setminus \{i\}} \rv_j\\
	\Leftrightarrow \; & \D \sum_{i \in C} x_i^C \left(\rv_i - \frac{1}{|C| - 1} \sum_{j \in C \setminus\{i\}} \rv_j\right) &\ge& 0 \\
	\end{array} 
	\end{equation}
	
	Similarly, $\receiver$ should not have an incentive to accept a candidate upon the event of $\sigma_t = \NOHIRE$. We again assume that the probability of this event is positive. The probability that upon a \NOHIRE\ signal the sender sees candidate $i$ is $p^{nh}_i = \frac{1-x_i^C}{\sum_{j \in C} (1-x_j^C)}$. Upon deviation and hiring the candidate, $\receiver$ gets utility $\sum_{i \in C} p^{nh}_i \rv_i$. Upon following the signal, $\receiver$ gets a utility of $\sum_{i \in C} p^{nh}_i u^{\receiver}_{C \setminus \{i\}}$. Hence, persuasiveness requires
	
	\begin{equation}
	\label{eq:optRev2}
	\begin{array}{rrcl}
	&  \D\sum_{i \in C} p^{nh}_i u^{\receiver}_{C \setminus\{i\}} &\ge& \D \sum_{i \in C}p^{nh}_i \cdot \rv_i\\
	\Leftrightarrow \; &  \D\sum_{i \in C} (1-x_i^C) u^{\receiver}_{C \setminus\{i\}} &\ge& \D \sum_{i \in C} (1-x_i^C) \cdot \rv_i
	\end{array} 
	\end{equation}
	
	\begin{claim}
		\label{claim:redundant}
		Constraint \eqref{eq:optRev2} is redundant and implied by~\eqref{eq:optRev}.
	\end{claim}
	
	We show the claim below. Thus, we get the following family of linear programs to determine $u^{\sender}_C$:
	\begin{maxi!}{x}{\frac{1}{|C|} \sum_{i \in C} (x_i^{C} \sv_i + (1-x_i^{C})\cdot u^{\sender}_{C \setminus \{i\}})\label{LP:utilitynested}}{\label{LP:templatenested}}{}
		\addConstraint{\sum_{i \in C} x_i^{C} \left[\rv_i - \frac{1}{|C|-1} \sum_{j \in C\setminus \{i\}} \rv_j \right]}{\ge 0 \label{LP:constraintnested}}
		\addConstraint{x_i^C}{\in [0,1] \text{ for all } i \in C}
	\end{maxi!}
\end{proof}

\begin{proof}[Proof of Claim~\ref{claim:redundant}]
	Every solution $\bm{x}$ that satisfies \eqref{eq:optRev2} for some values for $u^{\receiver}_{C \setminus \{i\}}$ also fulfills it for pointwise larger values. This is clear since with expected utility in subsequent rounds becoming larger, $\receiver$ does not get incentive to deviate from a \NOHIRE\ signal and hire in round $t$. 
	
	Recall that $\receiver$ can always ignore all signals and hire a random candidate. Hence, clearly, $u^{\receiver}_{C \setminus \{i\}} \ge \frac{1}{|C|-1}\sum_{j\in C\setminus\{i\}} \rv_j$. The pointwise smallest values $u^{\receiver}_{C \setminus \{i\}} = \frac{1}{|C|-1}\sum_{j\in C\setminus\{i\}} \rv_j$ imply for the constraint
	
	\begin{equation}
	\label{eq:optRev3}
	\begin{array}{rrcl}
	&  \D\sum_{i \in C} (1-x_i^C) \cdot \frac{1}{|C|-1}\sum_{j\in C\setminus\{i\}} \rv_j &\ge& \D \sum_{i \in C} (1-x_i^C) \cdot \rv_i \\
	\Leftrightarrow \; &  \D\sum_{i \in C} x_i^C \left(\rv_i - \frac{1}{|C|-1} \sum_{j \in C\setminus\{i\}} \rv_j \right) &\ge& 0
	\end{array} 
	\end{equation}
\end{proof}

	What reduction of utility does $\sender$ suffer when rejected candidates are announced to $\receiver$? The following theorem shows an upper bound of $\frac 12$. 
	
	\begin{theorem}
        \label{theo:senderOrdinalReceiverCardinalKnownUtilDisclosureUB}
		For every $\varepsilon > 0$, there is an instance such that every persuasive mechanism in the basic scenario with disclosure guarantees at most a fraction of $\left(\frac{1}{2} + \varepsilon\right)$ of the optimum in the benchmark scenario. This holds for both cardinal as well as ordinal sender utility.
	\end{theorem}
	\begin{proof}
		Consider the following class of instances. Candidate $1$ has value $(\rv_1,\sv_1) = (\frac{n-2}{n-1},1)$, candidate $2$ has value $(\rv_2,\sv_2) = (0,0)$, and candidates $i = 3,\ldots,n$ have values $(\rv_i,\sv_i) = (1,0)$. Hence, this can be seen as an instance for ordinal as well as cardinal sender utility. In the remainder of the proof, we will only use ``expected utility''.
		
		In the benchmark case, the Pareto mechanism would pick $a = b = 1$ and signal \HIRE\ if and only if candidate 1 arrives. This would yield an expected utility of 1 for $\sender$ and $\frac{n-2}{n-1}$ for $\receiver$.
		
		For the announcement case, constraint \eqref{LP:constraintnested} implies that $x_2^C = 0$ is optimal for all $C \ni 2$, i.e., an optimal mechanism never recommends to hire 2. $\sender$ has no value for candidates other than 1, so w.l.o.g.\  the optimal mechanism will always try to recommend 1 with $x_1^C = 1$ for every $C \ni 1$. For every subset $C \not\ni 1$, the utility is 0, and we can assume that $x_i^C = 1$ for all $i \in C, i \ne 2$. 
		
		Let $k = \sqrt{n}$, and consider the event $E$ that the following two conditions are satisfied: (1) at least one of candidates $\{1,2\}$ arrives within the first $n-k$ rounds, and (2) candidate 2 arrives before candidate 1. Suppose we draw two arrival rounds for candidates 1 and 2 uniformly at random. The probability that we draw both rounds from the last $k$ rounds is $p = \frac{k}{n}\cdot\frac{k-1}{n-1}$. For any set of two rounds drawn uniformly at random, the probability that 2 arrives in the earlier one is 1/2. Thus,
		$$\Pr[E] = \frac{1}{2} - \frac{k(k-1)}{2n(n-1)}\enspace.$$ 
		Suppose candidate 2 arrives in round $t$. Since 2 is always rejected, consider any subsequent round $t+1$. The set $C$ of remaining candidates consists only of candidate 1 and a subset of candidates $3,\ldots,n$. Since the latter are symmetric, the optimal mechanism uses the same value $x_i^C = x_j^C = x_{-1}^C$ for all candidates $i,j \in C \setminus \{1\}$. The optimal mechanism needs to satisfy \eqref{LP:constraintnested}
		$$ (|C|-1) \cdot x_{-1}^C \cdot \left(1 - \frac{1}{|C|-1}  \cdot \left(|C|-2 + \frac{n-2}{n-1}\right)\right) + x_1^C \cdot \left(\frac{n-2}{n-1} - \frac{1}{|C|-1} \cdot (|C|-1)\right) \ge 0\enspace.$$
		Using the assumption that $x_1^C = 1$ and solving for $x_{-1}^C$ implies
		$$ x_{-1}^C \ge 1\enspace.$$
		Hence, after candidate 2 is rejected in round $t$, there is an optimal mechanism that signals \HIRE\ with probability 1 in round $t+1$. The probability that candidate 1 is hired in this round is only $1/|C| \le 1/{(k-1)}$. Note that this upper bound holds even conditioned on event $E$, since the round in which 1 arrives remains uniformly distributed among the remaining ones. Thus, the overall expected utility of the optimal mechanism in the disclosure case is upper bounded by
		$$ (1-\Pr[E])\cdot 1 + \Pr[E] \cdot \frac{1}{k-1} \le \frac{1}{2} + \frac{k(k-1)}{2n(n-1)} + \frac{1}{2(k-1)} - \frac{k}{2n(n-1)} = \frac{1}{2} + o(1)$$
		since $k = \sqrt{n}$.
	\end{proof}
	
	We now turn to our polynomial-time approximation mechanism for this scenario. It is another variant of the Pareto mechanism, termed Shrinking Pareto mechanism (see Algorithm~\ref{algo:ShrinkingPareto}). The Shrinking Pareto mechanism applies the Pareto procedure adaptively to the set of remaining candidates $R_t$. The mechanism signals $\HIRE$ in round $t$ if and only if the candidate chosen by the Pareto procedure on $R_t$ is the one arriving in round $t$.
	\begin{algorithm}[t]
		\caption{\label{algo:ShrinkingPareto} Shrinking Pareto Mechanism}
		\DontPrintSemicolon
		\KwIn{Pairs of candidate values $(\rv_i, \sv_i)_{i \in [n]}$ }
		$R_1 \leftarrow [n]$ \;
		\For{$t=1$ to $n$}{$c_t \leftarrow $ candidate chosen by Pareto procedure on the set $R_t$ \\
			\lIf{$c_t = \stateON_t$}{Signal $\HIRE$ and end mechanism}
			\lElse{$R_{t+1} \leftarrow R_t \setminus \{\theta_t\}$ and signal $\NOHIRE$}}
	\end{algorithm}
	
    We establish three properties of the Shrinking Pareto mechanism. We show that the mechanism is persuasive, it yields a $(1/3-o(1))$-approximation for cardinal sender utility, and a $(1/2 - o(1))$-approximation for ordinal sender utility. 

	\begin{lemma}\label{lem:shrinkingParetoIC}
		The Shrinking Pareto mechanism is persuasive in the basic scenario with disclosure.
	\end{lemma}

\begin{proof}
	Recall from our analysis above that the probabilities for a signal $\sigma_t = \HIRE$ resulting from the Pareto procedure are chosen as a linear combination, which yields at least an expected utility of $\mu^{\receiver}$ to $\receiver$. Formally, consider the decision of $\receiver$ in round $t$ conditioned on a subset $R_t$ of remaining candidates. Suppose the mechanism sends \HIRE. Then, the expected utility for $\receiver$ is
	$$ \Ex[ \rv(\stateON_t) \mid \sigma_t = \HIRE \wedge R_t ] \ge \frac{1}{|R_t|} \sum_{i \in R_t} \rv_i\enspace.$$
	We use $x_i^{R_t} = \Pr[\sigma_t = \HIRE \mid R_t \wedge (\stateON_t = i \in R_t)]$, i.e., the probability that the set $R_t$ remains to arrive in rounds $t,\ldots,n$ and $i \in R_t$ arrives in round $t$. For the Pareto procedure we have $\sum_{i \in R_t} x_i^{R_t} = 1$, thus
	$$ \frac{1}{|R_t|} \sum_{i \in R_t} \rv_i = \frac{1}{|R_t|} \sum_{i \in R_t} x_i^{R_t} \sum_{j \in R_t} \rv_j = \frac{1}{|R_t|} \sum_{i \in R_t} x_i^{R_t} \sum_{j \in R_t \setminus\{i\}} \rv_j + \frac{1}{|R_t|}\sum_{i \in R_t} x_i^{R_t} \rv_i\enspace.$$
	Therefore
	$$ \sum_{i \in R_t} x_i^{R_t} \rv_i \ge \frac{1}{|R_t|-1} \sum_{i \in R_t} x_i^{R_t} \sum_{j \in R_t \setminus \{i\}} \rv_j\enspace.$$
	The signaling probabilities $x_i^{R_t}$ of the Shrinking Pareto mechanism represent a feasible solution for every linear program~\eqref{LP:templatenested}. Hence, the mechanism is persuasive.
\end{proof}

    The following theorem is our main result for cardinal sender utility. The Shrinking Pareto mechanism guarantees a constant-factor approximation of the optimal utility achieved in the benchmark scenario. It is an interesting open question to prove a tight bound in this scenario.

    \begin{theorem}\label{theo:senderCardinalReceiverCardinalKnownUtilDisclosure}
	   In the basic scenario with disclosure, the Shrinking Pareto mechanism scenario obtains a $\left(\frac{1}{3} - o(1)\right)$-approximation of the optimum in the corresponding benchmark scenario.
    \end{theorem}

    \begin{proof}
	Consider an instance of the benchmark scenario with $n$ candidates. We use $\opt$ to denote the optimal expected utility of $\sender$ achieved by the Pareto mechanism. Consider the same instance in the public announcement scenario, and let $\ap$ be the expected utility of $\sender$ achieved by the Shrinking Pareto mechanism.
	
	We denote by $\stateON_t$ the candidate arriving in round $t$. Consider the first round $t=1$ and the application of the Pareto procedure. Obviously, $\sigma_1 = \HIRE$ with probability $\frac{1}{n}$, and the expected utility $\Ex[\sv_{\stateON_1} \mid \sigma_1 = \HIRE] = \opt$. 
	
	Now suppose $\sigma_1 = \NOHIRE$. For every subset $S \subseteq [n]$, let $\opt_{-S}$ and $\ap_{-S}$ be the expected utility of $\sender$ achieved by the (Shrinking) Pareto mechanism on the benchmark (disclosure) instance with candidates $[n] \setminus S$, respectively. For the expected value of the Shrinking Pareto mechanism we can derive the following recursive lower bound
	\begin{align*}
	&\frac{1}{n} \sum_{i \in [n]} \Pr[\stateON_1 = i, \sigma_1 = \HIRE] \cdot \sv_i + \Pr[\stateON_1 = i, \sigma_1 = \NOHIRE] \cdot \ap_{-\{i\}} \\
	&= \frac{\opt}{n} + \frac{1}{n} \left((1-\alpha)\ap_{-\{a\}} + \alpha \ap_{-\{b\}} + \sum_{i \neq a,b} \ap_{-\{i\}} \right)\\
	&\ge \frac{\opt}{n} + \frac{1}{n} \sum_{i \neq a,b} \ap_{-\{i\}}\\
	&\ge \frac{\opt}{n} + \frac{1}{n} \sum_{i \neq a,b} \left(\frac{1}{n-1} \sum_{j \in [n] \setminus\{i\}} \Pr[\stateON_2 = j, \sigma_2 = \HIRE] \cdot \sv_j + \Pr[\stateON_2 = j, \sigma_2 = \NOHIRE] \cdot \ap_{-\{i,j\}}\right)\\
	&\ge \frac{\opt}{n} + \frac{1}{n} \left(\sum_{i \neq a,b} \frac{\opt_{-i}}{n-1} + \frac{1}{n-1} \sum_{i \in [n]\setminus\{a,b\}} \sum_{j \in [n]\setminus\{i\}} \Pr[\stateON_2 = j,  \sigma_2 = \NOHIRE] \cdot \ap_{-\{i,j\}} 
	\right)\\
	&\ge \frac{\opt}{n} + \sum_{i \neq a,b} \frac{\opt_{-\{i\}}}{n(n-1)} + \sum_{i \in [n]\setminus\{a,b\}} \sum_{\substack{j \in [n]\setminus\{i\} \\ j \neq a_{-i},b_{-i}}} \frac{\opt_{-\{i,j\}}}{n(n-1)(n-2)} + \ldots 
	\end{align*}
	where $a_{-i}, b_{-i}$ are the candidates $a$ and $b$ identified by the Pareto procedure when applied to the candidate set $[n] \setminus \{i\}$. Applying the lower bound on $\sum_{i \neq a,b} \opt_{-i}$ from Lemma~\ref{lem:lossPerRoundCard} repeatedly in the formula above, we obtain
	\begin{align*}
	&\frac{\opt}{n} + \sum_{i \neq a,b} \frac{\opt_{-\{i\}}}{n(n-1)} + \sum_{i \in [n]\setminus\{a,b\}} \sum_{\substack{j \in [n]\setminus\{i\} \\ j \neq a_{-i},b_{-i}}} \frac{\opt_{-\{i,j\}}}{n(n-1)(n-2)} + \ldots \\
	&\ge \frac{\opt}{n} + \frac{n-3}{n(n-1)}\opt + \sum_{i \in [n]\setminus\{a,b\}} \frac{n-4}{n(n-1)(n-2)} \opt_{-\{i\}} + \ldots\\
	&\ge \opt \cdot \frac{1}{n(n-1)(n-2)} \left(\sum_{i=1}^{n-3} (n-i)(n-i-1)\right)\\
	&= \opt \cdot  \left(\frac{1}{3} - \frac{2}{n^3 - 3n^2 + 2n}\right)\enspace.
	\end{align*}
    \end{proof}	

    For the ordinal case, we show a bound that asymptotically matches the upper bound shown in Theorem~\ref{theo:senderOrdinalReceiverCardinalKnownUtilDisclosureUB}. In contrast to the cardinal case, this implies that the approximation guarantee of our mechanism is asymptotically optimal. The proof is conceptually similar to the previous theorem, but instead of Lemma~\ref{lem:lossPerRoundCard} relies on the improved guarantee in Lemma~\ref{lem:lossPerRoundOrd}. 

	\begin{theorem}\label{theo:senderOrdinalReceiverCardinalKnownUtilDisclosure}
		The Shrinking Pareto mechanism yields a success probability of at least $\left(\frac{1}{2}-o(1)\right)$ in the basic scenario with disclosure. 
	\end{theorem}
	\begin{proof}
		Let $SP$ be the success probability from the Shrinking Pareto mechanism.
		We show that
		\[ SP \ge \opt\cdot \left(\frac{1}{2} - \frac{1}{2n}\right)\enspace.\]
		The utility obtained in a round depends on the current candidate. With probability $\frac{1}{n}$, the current candidate is $a$. The signal $\HIRE$ is sent with probability $\alpha = \opt$. With probability $1-\opt$, $\NOHIRE$ is sent.
		Analogously, if the current candidate is $b$, $\NOHIRE$ is sent with probability $\opt$ while the signal is $\HIRE$ with probability $1-\opt$. Any other candidate is guaranteed to get a signal $\NOHIRE$. In the following, we denote by $SP_{-M}$ for any subset $M$ of candidates the expected utility from the Shrinking Pareto mechanism for the instance $[n] \setminus M$. We apply Lemma~\ref{lem:lossPerRoundOrd} recursively and obtain
		\begin{align*}
		SP &= \frac{1}{n}\opt + \frac{1}{n}\opt \cdot SP_{-\{c_\receiver\}} + \frac{1}{n} \sum_{i \ne c_\sender,c_\receiver} SP_{-\{i\}} \\
		&= \frac{\opt}{n} + \frac{\opt}{n} \left[ \frac{1}{n-1}\opt_{-\{c_\receiver\}} + \frac{1}{n-1}\opt_{-\{c_\receiver\}}\cdot SP_{-\{c_\receiver,c_\receiver'\}} + \frac{1}{n-1} \sum_{i \ne c_\sender,c_\receiver,c_\receiver'} SP_{-\{c_\receiver,i\}} \right] \\*
		&\hspace{25pt} + \frac{1}{n} \sum_{i \ne c_\sender,c_\receiver} \left[\frac{1}{n-1}\opt_{-\{i\}} + \frac{1}{n-1} \opt_{-\{i\}} SP_{-\{i,c_\receiver\}} + \frac{1}{n-1} \sum_{j \ne c_\sender,c_\receiver,i} SP_{-\{i,j\}}\right] \\
		&= \frac{\opt}{n} + \underbrace{\frac{\opt}{n} \cdot \frac{\opt_{-\{c_\receiver\}}}{n-1} + \frac{1}{n}\sum_{i \ne c_\sender,c_\receiver} \frac{\opt_{-\{i\}}}{n-1}}_{\ge \frac{1}{n(n-1)} \cdot \opt \cdot \left(n-2-\frac{1}{n-1}\right)} + \frac{\opt}{n}\cdot\frac{\opt_{-\{c_\receiver\}}}{n-1} \cdot SP_{-\{c_\receiver,c_\receiver'\}} \\
		&\hspace{25pt} + \frac{\opt}{n}\frac{1}{n-1} \sum_{i \ne c_\sender,c_\receiver,c'_\receiver} SP_{-\{c_\receiver,i\}} + \frac{1}{n} \sum_{i \ne c_\sender,c_\receiver} \frac{1}{n-1} \opt_{-\{i\}} \cdot SP_{-\{i,c_\receiver\}} \\
		&\hspace{25pt} + \frac{1}{n}\frac{1}{n-1} \sum_{i \ne c_\sender,c_\receiver} \sum_{j \ne c_\sender,c_\receiver,i} SP_{-\{i,j\}} \\
		&\ge \frac{\opt}{n} + \frac{1}{n}\frac{1}{n-1} \cdot \left(n-2-\frac{1}{n-1}\right)\cdot \opt + \frac{\opt}{n}\cdot\frac{\opt_{-\{c_\receiver\}}}{n-1} \cdot SP_{-\{c_\receiver,c_\receiver'\}} \\
		&\ge \frac{\opt}{n}\left[1 + \frac{1}{n-1} \left(n-2-\frac{1}{n-1}\right) + \frac{1}{n-1}\frac{1}{n-2} \cdot \left( n-2-\frac{1}{n-1}\right)\left( n-3-\frac{1}{n-2}\right) + \dots \right] +  \dots \\
		&\ge \opt\left(\frac{1}{2} - \frac{1}{2n}\right) \enspace. 
		\end{align*}
	\end{proof}

\subsection{Secretary Scenario with Disclosure}
    Finally, consider the secretary scenario with disclosure and cardinal receiver utility. In this scenario, $\sender$ does not have any information on the valuation pairs. Moreover, the rejected candidates are revealed to $\receiver$. Initially, $\receiver$ has the same information as usual (knows all valuation pairs, knows they arrive in random order, knows mechanism $\phi$). At the end of each round with a rejection, $\receiver$ also learns the valuation pair of the candidate she just rejected.

    Obviously, $\sender$ can apply the \emph{trivial mechanism}: Recommending a candidate chosen uniformly at random. This mechanism is persuasive. It achieves a success probability $1/n$ and an expected utility of $1/n \cdot \sv_{\max}$, a trivial lower bound.

    We complement this insight with a strong upper bound for a sender with ordinal utility. The subsequent Corollary~\ref{cor:senderCardinalReceiverCardinalSecretaryDisclosure} shows a similar bound when the sender has cardinal utility.

    \begin{theorem}\label{theo:senderOrdinalReceiverCardinalSecretaryDisclosure}
	   In the secretary scenario with disclosure, there is no persuasive mechanism that guarantees $\sender$ a success probability greater than $\frac{2}{n} \cdot \opt$, where $\opt$ is the success probability in the corresponding benchmark instance.
    \end{theorem}

    The lower bound follows from the following two different instances. In both instances, there is one candidate $a$ with value-pair $(\rv_a,\sv_a) = (0,1)$ and $n-2$ indistinguishable candidates (each termed $c$) with value-pair $(1/2,0)$. In instance I, there is one additional candidate $b$ with value-pair $(1,0)$. In instance II, there is just another $c$-candidate with $(1/2,0)$.

    Hence, the only uncertainty for $\sender$ in this set of instances is about the existence of $b$. 

    Now consider an arbitrary persuasive mechanism. Clearly, we can assume that the mechanism issues exactly one $\HIRE$ signal in exactly one of the rounds $t \in [n]$. If this signal is not followed, it just issues $\NOHIRE$ signals. This can only increase the incentive for $\receiver$ to follow the $\HIRE$ signal. We always assume that the mechanism issues a $\HIRE$ signal in the last round if it has not done so before. This is just a standard assumption. Note that in the last round, the arriving candidate is perfectly known to $\receiver$.

Now consider round $t = 3,\ldots,n-1$. Suppose the mechanism has only sent $\NOHIRE$ signals so far. Based on the set of candidates $A_{t-1}$ that arrived up to and including round $t-1$ (and whose arrival is now known to both $\sender$ and $\receiver$), we distinguish eight cases.

\begin{center}
\renewcommand{\arraystretch}{1.5}
\begin{tabular}{l||c|c|c|c|c|c|c|c}
$\theta_t$                & $a$      & $a$      & $b$      & $b$      & $c$      & $c$      & $c$      & $c$\\ \hline
Candidates in $A_{t-1}$   & only $c$ & $b$, $c$ & only $c$ & $a$, $c$ & only $c$ & $a$, $c$ & $b$, $c$ & $a$, $b$, $c$ \\ \hline
$\Pr[\sigma_t = \HIRE \mid A_{t-1},\theta_t]$   & $p_{a,c}^t$ & $p_{a,b}^t$ & $p_{b,c}^t$ & $p_{b,a}^t$ & $p_{c,c}^t$ & $p_{c,a}^t$ & $p_{c,b}^t$ & 1
\end{tabular}
\end{center}

Note that, obviously, in the last case when both $a$ and $b$ arrived and were rejected, it is clear to both $\sender$ and $\receiver$ that the underlying instance is instance I and there are only candidates $c$ to come in the subsequent rounds. Hence, in this case we can w.l.o.g.\ assume that $\sender$ issues a $\HIRE$ signal directly. In round $t=2$, this case does not exist:

\begin{center}
\renewcommand{\arraystretch}{1.5}
\begin{tabular}{l||c|c|c|c|c|c|c}
$\theta_2$              & $a$ & $a$ & $b$ & $b$ & $c$ & $c$ & $c$ \\ \hline
Candidate $\theta_1$    & $c$ & $b$ & $c$ & $a$ & $c$ & $a$ & $b$ \\ \hline
$\Pr[\sigma_2 = \HIRE \mid \theta_1, \theta_2]$ & $p_{a,c}^2$ & $p_{a,b}^2$ & $p_{b,c}^2$ & $p_{b,a}^2$ & $p_{c,c}^2$ & $p_{c,a}^2$ & $p_{c,b}^2$
\end{tabular}
\end{center}
 
Finally, for round $t=1$, there are only three cases depending on whether $\theta_1$ is $a$, $b$ or $c$. We denote the probabilities for a $\HIRE$ signal by $p_a^1$, $p_b^1$ and $p_c^1$, respectively.

In the next lemma, we show necessary constraints for persuasiveness.

\begin{lemma} \label{lem:IncentiveConstraints}
    For every mechanism that is persuasive in both instances I and II, it must hold for round $t=1$
\[
    p_b^1 \ge p_a^1 \hspace{1cm} \text{and} \hspace{1cm} p_c^1 \ge p_a^1
\]
and for every round $t=2,\ldots,n-1$
\[
    p_{b,c}^t \ge p_{a,c}^t \hspace{1cm} p_{b,a}^t \ge p_{c,a}^t  \hspace{1cm}  p_{c,b}^t \ge p_{a,b}^t     \hspace{1cm} p_{c,c}^t \ge p_{a,c}^t\enspace.
\]
\end{lemma}

\begin{proof}
	Suppose there is a $\HIRE$ signal in round $t$. $\receiver$ must find it in her interest to follow the signal. More precisely, in what follows we condition on $\sigma_t = \HIRE$. Now, for persuasiveness the expected utility from $\theta_t$ must exceed the expected utility from $\theta_{t+1}$.
	
	First, consider round $t=1$:
	\begin{description}
		\item[Instance I:]
		We must have $\Ex[\rv_{\theta_1} \mid \sigma_1 = \HIRE] \; \ge \; \Ex[\rv_{\theta_2} \mid \sigma_1 = \HIRE]$, which implies
		\begin{align*}
		&p_a^1 \cdot 0 + p_b^1 \cdot 1 + (n-2) \cdot p_c^1 \cdot 1/2\\
		&\quad\ge \quad p_a^1 \cdot \frac{(n-2)\cdot 1/2 + 1}{n-1} + p_b^1 \cdot \frac{(n-2)\cdot 1/2}{n-1} + (n-2)\cdot p_c^1 \cdot \frac{(n-3)\cdot 1/2 + 1}{n-1}
		\end{align*}
		The constraint becomes $p_b^1 \ge p_a^1$.
		
		\item[Instance II:]
		We must have $\Ex[\rv_{\theta_1} \mid \sigma_1 = \HIRE] \; \ge \; \Ex[\rv_{\theta_2} \mid \sigma_1 = \HIRE]$, which implies
		\begin{align*}
		&p_a^1 \cdot 0 + (n-1) \cdot p_c^1 \cdot 1/2 \quad\ge \quad p_a^1 \cdot 1/2 + (n-1)\cdot p_c^1 \cdot \frac{(n-2)\cdot 1/2}{n-1}
		\end{align*}
		The constraint becomes $p_c^1 \ge p_a^1$.
	\end{description}
	
	Hence, the mechanism must be more likely to signal \HIRE\ in round 1 for each $c$ and $b$ than for $a$. Technically, these two constraints also hold for round $t=n$, since the probability to signal \HIRE\ is 1 for every candidate. The two constraints also emerge in the remaining rounds $t=2,\ldots,n-1$:
	
	\begin{description}
		\item[Instance I, $A_{t-1}$ contains only $c$:] $\Ex[\rv_{\theta_t} \mid \sigma_t = \HIRE] \; \ge \; \Ex[\rv_{\theta_{t+1}} \mid \sigma_t = \HIRE]$ implies
		\begin{align*}
		& p_{a,c}^t \cdot 0 + p_{b,c}^t \cdot 1 + (n-t-1)\cdot p_{c,c}^t \cdot 1/2\\
		&\quad\ge p_{a,c}^t \cdot \frac{(n-t-1)\cdot 1/2 + 1}{n-t} + p_{b,c}^t \cdot \frac{(n-t-1)\cdot 1/2}{n-t} + (n-t-1) \cdot p_{c,c}^t \cdot \frac{(n-t-2)\cdot 1/2 + 1}{n-t}
		\end{align*}
		Observe that the $p_{c,c}^t$ terms cancel, and the constraint becomes $p_{b,c}^t \ge p_{a,c}^t$.
		
		\item[Instance I, $A_{t-1}$ contains $a$:] $\Ex[\rv_{\theta_t} \mid \sigma_t = \HIRE] \; \ge \; \Ex[\rv_{\theta_{t+1}} \mid \sigma_t = \HIRE]$ implies
		\begin{align*}
		&p_{b,a}^t \cdot 1 + (n-t)\cdot p_{c,a}^t \cdot 1/2 \quad\ge \quad p_{b,a}^t \cdot 1/2 + (n-t) \cdot p_{c,a}^t \cdot \frac{(n-t-1)\cdot 1/2 + 1}{n-t}
		\end{align*}
		The constraint becomes $p_{b,a}^t \ge p_{c,a}^t$.
		
		\item[Instance I, $A_{t-1}$ contains $b$:] $\Ex[\rv_{\theta_t} \mid \sigma_t = \HIRE] \; \ge \; \Ex[\rv_{\theta_{t+1}} \mid \sigma_t = \HIRE]$ implies
		\begin{align*}
		&p_{a,b}^t \cdot 0 + (n-t)\cdot p_{c,b}^t \cdot 1/2 \quad\ge \quad p_{a,b}^t \cdot 1/2 + (n-t) \cdot p_{c,b}^t \cdot \frac{(n-t-1)\cdot 1/2}{n-t}
		\end{align*}
		The constraint becomes $p_{c,b}^t \ge p_{a,b}^t$.
		
		\item[Instance I, $A_{t-1}$ contains $a$, and $b$:] For $t=2$ this case does not occur. For $t\ge 3$, both players are aware that there are only $c$-candidates left. $\receiver$ follows any signaling strategy that guarantees a $\HIRE$ signal in the remaining rounds.
		
		\item[Instance II, $A_{t-1}$ contains only $c$:] $\Ex[\rv_{\theta_t} \mid \sigma_t = \HIRE] \; \ge \; \Ex[\rv_{\theta_{t+1}} \mid \sigma_t = \HIRE]$ implies
		\begin{align*}
		&p_{a,c}^t \cdot 0 + (n-t)\cdot p_{c,c}^t \cdot 1/2 \quad\ge \quad p_{a,c}^t \cdot 1/2 + (n-t) \cdot p_{c,c}^t \cdot \frac{(n-t-1)\cdot 1/2}{n-t}
		\end{align*}
		The constraint becomes $p_{c,c}^t \ge p_{a,c}^t$.
		
		\item[Instance II, $A_{t-1}$ contains $a$:]
		In this case $\receiver$ knows that there are only $c$-candidates left. $\receiver$ follows any signaling strategy that guarantees a $\HIRE$ signal in the remaining rounds.
	\end{description}
\end{proof}

The proof of the theorem now follows by analyzing the approximation ratio when applying such a mechanism in instance I.

\begin{proof}[Proof of Theorem~\ref{theo:senderOrdinalReceiverCardinalSecretaryDisclosure}]
	Consider instance I. Instance I in the benchmark scenario yields an optimal success probability of 1/2 (recommend $a$ or $b$ each with probability 1/2). Let us analyze the performance of a mechanism on instance I that satisfies Lemma~\ref{lem:IncentiveConstraints}. Such a mechanism is much more careful to recommend $a$. If candidate $b$ would not exist, there would be a lack of interest of $\receiver$ (that knows instance II is present) to accept a recommendation of $a$. This risk arising from the potential presence of instance II is captured by the constraints in Lemma~\ref{lem:IncentiveConstraints}.
	
	Formally, if the mechanism reaches round $n$ without a previous $\HIRE$ signal, then, obviously, it is optimal to signal $\HIRE$ in round $n$ with probability 1. We will prove by induction that there is an optimal mechanism $\phi$ with the following property. If $\phi$ reaches any round $t$ without a previous $\HIRE$ signal, it is optimal to hire in round $t$ with probability 1. Thus, there is an optimal mechanism that hires in round 1 with probability 1. Obviously, this mechanism has success probability $1/n$.
	
	Suppose the inductive assumption is true for rounds $t+1,\ldots,n$, i.e., if the mechanism reaches round $t+1$ without a previous $\HIRE$ signal, it is optimal to hire any candidate in that round with probability 1. 
	
	Now consider round $t \ge 2$.
	
	\begin{description}
		\item[$A_{t-1}$ contains only $c$:]
		If $a$ arrives in round $t$, then the success probability is $p_{a,c}^t$. If $b$ arrives, the success probability is $(1-p_{b,c}^t)\cdot\frac{1}{n-t}$, since by assumption the mechanism hires every candidate in the next round, which due to random arrival is $a$ with uniform probability. If $c$ arrives, the success probability is $(1-p_{c,c}^t)\cdot\frac{1}{n-t}$ by similar arguments. Overall, we want to maximize
		\[
		p_{a,c}^t + \frac{1-p_{b,c}^t}{n-t} + (n-t-1)\cdot \frac{1-p_{c,c}^t}{n-t}\enspace.
		\]
		Due to Lemma~\ref{lem:IncentiveConstraints} $p_{b,c}^t \ge p_{a,c}^t$ and $p_{c,c}^t \ge p_{a,c}^t$. Hence, the expression is maximized if both constraints hold with equality, in which case it becomes
		\[
		p_{a,c}^t + \frac{1-p_{a,c}^t}{n-t} + (n-t-1)\cdot \frac{1-p_{a,c}^t}{n-t} = 1\enspace,
		\]
		i.e., \emph{independent} of the value of $p_{a,c}^t$. Thus, $p_{a,c}^t = p_{b,c}^t = p_{c,c}^t = 1$ is an optimal choice.
		
		\item[$A_{t-1}$ contains $a$:]
		If $a$ has arrived and been rejected, the success probability is 0. Setting $p_{b,a} = p_{c,a} = 1$ is an optimal choice.
		
		\item[$A_{t-1}$ contains $b$:]
		If $a$ arrives in round $t$, then the success probability is $p_{a,b}^t$. If $c$ arrives, the success probability is $(1-p_{c,b}^t)\cdot\frac{1}{n-t}$, since by assumption the mechanism hires every candidate in the next round, which due to random arrival is $a$ with uniform probability. Overall, we want to maximize
		\[
		p_{a,b}^t + (n-t)\cdot \frac{1-p_{c,b}^t}{n-t}\enspace.
		\]
		Due to Lemma~\ref{lem:IncentiveConstraints} $p_{c,b}^t \ge p_{a,b}^t$. Hence, the expression is maximized when the constraints holds with equality, in which case it becomes
		\[
		p_{a,b}^t + (n-t)\cdot \frac{1-p_{a,b}^t}{n-t} = 1\enspace.
		\]
		Thus, $p_{a,b}^t = p_{c,b}^t = 1$ is an optimal choice.
		
		\item[$A_{t-1}$ contains $a$, and $b$:] If $a$ has arrived and been rejected, the success probability is 0. We already observed above that in this case we can directly hire in round $t$ with probability 1.
	\end{description}
	
	For round $t=1$ we make a final similar observation. The success probability is $p_a^1$ if $a$ arrives, $(1-p_b^1) \cdot \frac{1}{n-1}$ if $b$ arrives and $(1-p_c^1) \cdot \frac{1}{n-1}$ if $c$ arrives. Overall, we want to maximize
	\[
	p_a^1 + \frac{1-p_b^1}{n-1} + (n-2)\cdot \frac{1-p_c^1}{n-1}
	\]
	Due to Lemma~\ref{lem:IncentiveConstraints} $p_b^1 \ge p_a^1$ and $p_c^1 \ge p_a^1$. Hence, the expression is maximized when both constraints hold with equality, in which case it becomes
	\[
	p_a^1 + (n-1)\cdot \frac{1-p_a^1}{n-1} = 1\enspace.
	\]
	Thus, $p_a^1 = p_b^1 = p_c^1 = 1$ is an optimal choice.
	
	This proves the theorem.
\end{proof}

The proof can be applied literally by simply replacing ``success probability'' with ``expected utility'' for $\sender$, since in both instances I and II the best candidate $a$ is the only one that yields positive utility for $\sender$. This implies the following corollary.

\begin{corollary}\label{cor:senderCardinalReceiverCardinalSecretaryDisclosure}
	There is no persuasive mechanism that guarantees $\sender$ an expected utility of more than $\frac{2}{n}\cdot\opt$, where $\opt$ is the optimal expected utility in the corresponding benchmark instance.
\end{corollary}

\subsection*{Acknowledgements}
Niklas Hahn and Martin Hoefer are funded by the German-Israel Foundation grant I-1419-118.4/2017.
Rann Smorodinsky is funded by the joint United States-Israel Binational Science Foundation and National Science Foundation grant 2016734,
German-Israel Foundation grant I-1419-118.4/2017, Israel Ministry of Science and Technology grant 19400214, Technion VPR grants, and the Bernard M. Gordon Center for Systems Engineering at the Technion.

\bibliographystyle{plain}
\bibliography{literature,martin}

\clearpage
\appendix










\section{Ordinal Utility for $\receiver$}
\label{sec:ordinal}

Let us now consider the problem when $\receiver$ is only interested in hiring her (unique) best candidate, i.e., the case of ordinal receiver utility. 

\subsection{Basic Scenario without Disclosure}
In this scenario, the set of valuation pairs is known a-priori to both parties. While $\sender$ observes the candidates, $\receiver$ only observes the signals sent by $\sender$. $\receiver$ is indifferent about all candidates which are not her best. As such, for a persuasive mechanism for $\sender$, we can restrict to signal $\HIRE$ for exactly one of the two optimal candidates $c_\receiver$ and $c_\sender$. 

In the Elementary mechanism, we decide upfront whether to signal $\HIRE$ for either $c_\receiver$ or $c_\sender$: Draw $x \sim Unif[0,1]$. If $x \le 1/n$, signal $\HIRE$ upon arrival of $c_\receiver$, otherwise, signal $\HIRE$ upon arrival of $c_\sender$. Signal $\NOHIRE$ for any other candidate.

\begin{proposition}\label{prop:receiverOrdinalKnownUtilNoDisclosure}
	The Elementary mechanism is persuasive in the basic scenario without disclosure. It yields a success probability of $(1-o(1))$ and an expected utility of $(1-o(1))\cdot \sv_{\max}$ for $\sender$.
\end{proposition}
\begin{proof}
	First, we prove persuasiveness.
	
	In round $t$ with $\sigma_t = \HIRE$ we have $\Pr[\theta_t = c_\receiver] = \frac{1/n}{1/n+(n-1)/n} = 1/n$. Suppose $\receiver$ deviates to hire in a later round $r > t$. Since $\sender$ will not signal $\HIRE$ again, $\receiver$ must choose $r$ without additional information. Then $\Pr[\theta_r = c_\receiver] = (1-1/n) \cdot 1/(n-1) = 1/n$. Thus, it is optimal for $\receiver$ to hire $\theta_t$.
	
	If $\sigma_{t'} = \NOHIRE$ for all $t' \in [t]$, then $\Pr[\theta_t = c_{\receiver}] = \frac{\frac{n-1}{n}}{n-2+\frac{1}{n}+\frac{n-1}{n}} = \frac{1}{n}$. Hence, it is optimal for $c_{\receiver}$ to wait for the round with signal $\HIRE$.
	
	Additionally, the Elementary mechanism clearly implies the sender-optimal candidate $c_\sender$ is hired with probability $1-o(1)$ and yields an expected utility of $(1-o(1))\sv_{\max}$ for $\sender$. 
\end{proof}

\subsection{Secretary Scenario without Disclosure}	
In this case, we assume that $\receiver$ has the same information as in the basic scenario. Now the valuation pairs are unknown to $\sender$. In this scenario, she can use the following mechanism that relies on the classic secretary algorithm due to Dynkin~\cite{Dynkin63}. 

The Simple Secretary mechanism decides once in the beginning whether to run the classic algorithm based either on $\sv$-values or $\rv$-values. In the classic algorithm, we sample the first $s = \lfloor n/e \rfloor$ candidates. Then we signal $\HIRE$ for the first candidate that is the best one so far (in terms of either $\sv$- or $\rv$-values, depending on the variant that is used). The classic algorithm hires the best candidate with probability $1/e - o(1)$. For our mechanism, with probability $1-e/n$, we run the classic algorithm using the $\sv$-values of $\sender$, with probability $e/n$ the $\rv$-values of $\receiver$. Note that $e/n = o(1)$, so as the number of candidates increases, the mechanism tends to run the classic algorithm almost exclusively in the sender-optimal version. By assumption, if the mechanism decided to signal $\sigma_{t} = \NOHIRE$ in all rounds $t = 1,\ldots,n-1$, it sets $\sigma_n = \HIRE$.

\begin{theorem}\label{theo:receiverOrdinalSecretaryNoDisclosure}
	The Simple Secretary mechanism is persuasive in the secretary scenario without disclosure. It yields a success probability of $1/e-o(1)$ and an expected utility of $(1/e-o(1))\cdot \sv_{\max}$ for $\sender$.
\end{theorem}
\begin{proof}
	We begin the proof by showing persuasiveness.
	
	Due to random order and the symmetric structure of the secretary algorithm, it produces the same distribution of $\HIRE$ signals, both for the sender-optimal or the receiver-optimal version of the secretary algorithm. In particular, for every round $t > \lfloor n/e \rfloor$, and for any set $A_t$ of arrived candidates in the rounds $1,\ldots,t$, the mechanism sends a hire signal in round $t$ when the best candidate from $A_t$ arrives in round $t$ and the second best arrived in the sample phase. Hence, $\Pr[\sigma_t = \HIRE \mid A_t] = \frac{1}{t} \cdot \frac{s}{t-1}$, which is the same for every set $A_t$ and for both the sender-optimal or receiver-optimal variant. There is no disclosure, so $\receiver$ cannot distinguish which variant is used.
	
	We first show persuasiveness for negatively correlated utilities, i.e., when the $i$-th best candidate for $\sender$ is the $(n-i+1)$-th best candidate for $\receiver$.
	
	Consider a round $t \in [n-1]$.
	\begin{description}
		\item[Case $\sigma_t = \HIRE$:] With probability $e/n$ the receiver-optimal version is run and the candidate is the best so far in terms of $\rv$. The probability that this candidate is actually $c_\receiver$ is $t/n$. If the sender-optimal version is run, there is no $\HIRE$ signal for $c_\receiver$. Hence, overall in this case $\Pr[\theta_t = c_\receiver] = et/n^2$.
		
		If $\receiver$ decides to deviate to a later round, she does not receive additional information from $\sender$. Due to random order arrival and the independence of the $\HIRE$ signal of the set of arrived candidates, in any given round $r > t$ we have $\Pr[\theta_r = c_\receiver] = \frac{n-t}{n} \cdot \frac{1}{n-t} = \frac{1}{n}$. Since $t \ge s+1 > n/e$ we have $et/n^2 > 1/n$. Hence, following the $\HIRE$ signal is optimal.
		
		\item[Case $t > s$ and $\sigma_i = \NOHIRE$ for all $i \le t$:] With probability $e/n$ the receiver-optimal version is run and the current candidate in round $t$ is not the best so far in terms of $\rv$. Then it cannot be $c_\receiver$. With probability $1-e/n$ the sender-optimal version is run and the candidate in round $t$ is not the best so far in terms of $\sv$. Then, $c_\receiver$ has probability $t/n$ to be among the set $A_t$ of the first $t$ candidates. Given that $c_\receiver \in A_t$, it arrives in round $t$ with probability $\frac{1}{t-1}$ -- there have been only $\NOHIRE$ signals, so the best candidate from $A_t$ must be in the sample phase. Hence, overall in this case $\Pr[\theta_t = c_\receiver] = \left(1-\frac{e}{n}\right) \cdot \frac{t}{n(t-1)} = \frac{(n-e)t}{n^2(t-1)}$.
		
		If instead $\receiver$ follows the mechanism, then in round $r > t$ the candidate is best so far with probability $\frac{1}{r} \cdot \frac{s}{r-1} = \frac{s}{r(r-1)}$. Note that in this case, we condition on $\NOHIRE$ signals in all rounds $\le t$. Hence, candidate $\theta_r$ is best so far with probability $\frac{s/(r(r-1))}{s/t}$. This candidate is $c_\receiver$ with probability $er/n^2$ as noted above. 
		
		Moreover, in the last round, there is always a $\HIRE$ signal. The probability that no candidate after the sample phase is ever the best so far is $\frac{s}{n}$. Since we condition on $\NOHIRE$ signals in all rounds $\le t$, this gives a probability of $\frac{s/n}{s/t}$. This candidate can be $c_\receiver$ only when we run the sender-optimal variant. Hence, the probability that this candidate is $c_\receiver$ is $\left(1-\frac{e}{n}\right) \cdot \frac{1}{n-1}$. Hence, in this case, following the mechanism $\receiver$ obtains $c_\receiver$ with probability
		\begin{equation}
		\label{eq:conditionedProp}
		\sum_{r=t+1}^n \frac{s/(r(r-1))}{s/t} \cdot \frac{er}{n^2} + \frac{s/n}{s/t} \cdot \left(1-\frac{e}{n}\right)\cdot \frac{1}{n-1} = \frac{t}{n^2} \left( \frac{n-e}{n-1} + \sum_{r=t+1}^n \frac{e}{r-1} \right)
		\end{equation}
		Thus, deviation by hiring in round $t$ is unprofitable if
		\[ 
		\frac{n-e}{t-1} \le \frac{n-e}{n-1} + e\cdot \sum_{r=t+1}^n \frac{1}{r-1} 			\]
		or, equivalently
		\[ 
		\frac{1}{t-1} - \frac{1}{n-1}	\le \frac{e}{n} \left(\frac{1}{t-1} - \frac{1}{n-1} + \sum_{r=t+1}^n \frac{1}{r-1} \right) 	\]
		and, thus,
		\begin{align}	
		\label{eq:crazy} 
		\frac{n-t}{(n-1)(t-1)\left(\sum_{r=t-1}^{n-2} \frac{1}{r} \right)} \le \frac{e}{n}\enspace.
		\end{align}
		
		\begin{claim}
			\label{cl:monotone}
			The lower bound in equation~\eqref{eq:crazy} is monotonically decreasing in $t$. Hence, the strongest lower bound arises in round $t = s+1$.
		\end{claim}
		
		\begin{proof}
			We want to prove that 
			$$
			\frac{n-t}{(n-1)(t-1)\left(\sum_{r=t-1}^{n-2} \frac{1}{r} \right)} \le
			\frac{n-t+1}{(n-1)(t-2)\left(\sum_{r=t-2}^{n-2} \frac{1}{r} \right)}
			$$
			We drop the terms $(n-1)$, bring the summations up, and cluster the non-summation terms on the right side. This yields
			$$\sum_{r=t-2}^{n-2} \frac{1}{r} \le \left(1 + \frac{n-1}{(n-t)(t-2)}\right) \cdot \left(\sum_{r=t-1}^{n-2} \frac{1}{r} \right) $$
			which implies
			$$ \frac{1}{t-2} \le \frac{n-1}{(n-t)(t-2)} \cdot \left(\sum_{r=t-1}^{n-2} \frac{1}{r} \right) $$
			and, hence
			$$ n-t \le (n-1) \left(\sum_{r=t-1}^{n-2} \frac{1}{r} \right)\enspace. $$
			This is true for $t = n-1$. Suppose it is true for some value of $t$, then it is also true for $t'= t-1$ -- the left-hand side increases by 1, the right-hand side increases by $\frac{n-1}{t-2} > 1$. This proves the claim.
		\end{proof}
		
		Using the claim, we can restrict attention to
		\begin{align*} 
		\frac{e}{n} &\ge \frac{n-s-1}{(n-1)s\left(\sum_{r=s}^{n-2} \frac{1}{r} \right)}
		\end{align*}
		Recall that $s = \lfloor n/e \rfloor$. For small $n \le 8$, the bound follows by direct inspection. For $n \ge 9$, we use the following inequality (cf \cite{Young91}), where $\gamma$ is the Euler-Mascheroni constant and $H_n$ the $n$-th harmonic number: 
		\[ \frac{1}{2(n+1)} < H_n - \ln n - \gamma < \frac{1}{2n} \]
		and bound
		\begin{equation}
		\label{eqn:harmonicnumber}
		\sum_{t=\left\lfloor \frac{n}{e} \right\rfloor}^{n-2} \frac{1}{t} = H_{n-2} - H_{\left\lfloor \frac{n}{e} \right\rfloor-1} \ge  \frac{1}{2(n-1)} + \ln(n-2) - \frac{1}{2(\left\lfloor \frac{n}{e}\right\rfloor-1)} - \ln \left(\left\lfloor \frac{n}{e} \right\rfloor -1\right).
		\end{equation}
		Further, we have $n \ge 3+2e$ and thus $2n-2-2e \ge n+1$. Hence, the following holds:
		\begin{align*}
		1 \quad & \ge \quad \frac{n+1}{2(n-(1+e))} 
		\quad = \quad \frac{n^2+n}{2n^2-2(1+e)n} \\
		& \ge \quad \frac{2n^2(\frac{1}{e}-\frac{1}{e^2})+n(e+\frac{2}{e}-3)}{2n^2-2(1+e)n+2e} 
		\quad = \quad \frac{2n(1-\frac{1}{e})}{2(n-1)e} - \frac{1}{2(n-1)} + \frac{1}{2(\frac{n}{e}-1)} \\
		& \ge \quad \frac{n(1-\frac{1}{e})}{(n-1)e} - \frac{1}{2(n-1)} + \frac{1}{2(\lfloor \frac{n}{e} \rfloor -1)}
		\end{align*}
		Combining this with (\ref{eqn:harmonicnumber}) and the fact that $n-2 \ge e \left(\left\lfloor \frac{n}{e} \right\rfloor -1\right)$ and hence $\ln \frac{n-2}{\left\lfloor \frac{n}{e} \right\rfloor-1} \ge 1$, we get 
		\[ \ln \frac{n-2}{\left\lfloor \frac{n}{e} \right\rfloor-1} + \frac{1}{2(n-1)} - \frac{1}{2\left(\left\lfloor \frac{n}{e} \right\rfloor -1\right)} \ge \frac{n\left(1-\frac{1}{e}\right)}{(n-1)e}
		\]
		and therefore 
		\[\frac{e}{n} \ge \frac{n-\left\lfloor \frac{n}{e} \right\rfloor - 1}{(n-1)\left\lfloor \frac{n}{e} \right\rfloor \sum_{t=\left\lfloor \frac{n}{e} \right\rfloor}^{n-2} \frac{1}{t}}\]
		as desired.
		
		\item[Case $t\le s$:]
		Finally, suppose \receiver\ deviates and hires a candidate during round $t \in [s]$. The probability of hiring $c_\receiver$ is $1/n$, no matter which variant of the algorithm is running. 
		
		If we want to bound the probability to hire $c_\receiver$ when following the mechanism, we can apply the bound in equation~\eqref{eq:conditionedProp}. However, since $t \le s$, we always have $\NOHIRE$ signals until round $t$. We remove the conditioning on having $\NOHIRE$ signals until round $t$, and the probability becomes 
		\[ 
		\sum_{r=s+1}^n \frac{s}{r(r-1)} \cdot \frac{er}{n^2} + \frac{s}{n}\cdot \frac{n-e}{n(n-1)} = \frac{s}{n^2} \left( \frac{n-e}{n-1} + \sum_{r=s+1}^n \frac{e}{r-1} \right)\enspace.
		\]
		Thus, accepting the candidate in round $t \le s$ is unprofitable if
		\begin{align*} 
		\frac{n}{s} &\le \frac{n-e}{n-1} + \sum_{r=s+1}^n \frac{e}{r-1}\enspace.
		\end{align*}
		This implies
		\begin{align*} 
		\frac{e}{n} \left(\sum_{r=s}^{n-2} \frac{1}{r} \right) &\ge \frac{1}{s} - \frac{1}{n-1}
		\end{align*}
		and, thus,
		\begin{align*} 
		\frac{e}{n} &\ge \frac{n-s-1}{(n-1)s\left(\sum_{r=s}^{n-2} \frac{1}{r} \right)}
		\end{align*}
		which we proved true in the previous case. 
	\end{description}
	Hence, the mechanism is persuasive if utilities are negatively correlated.
	
	If $c_\receiver$ is not the worst candidate for $\sender$, the following changes to the probabilities occur.
	
	Obviously, the probabilities of hiring $c_\receiver$ when running the receiver-optimal algorithm are the same. When running the sender-optimal algorithm, the probability of getting $c_\receiver$ upon a signal $\HIRE$ increases. Suppose the sender-rank of $c_\receiver$ is $x < n$. This means that there are $x-1$ candidates which $\sender$ prefers over $c_\receiver$. Thus, a signal $\HIRE$ for $c_\receiver$ implies that all $x-1$ candidates have to arrive after the current round. If they arrive during the sample phase, $c_\receiver$ will not get a $\HIRE$ signal, if they arrive between sample phase and the current round, they would have gotten the signal $\HIRE$ instead of $c_\receiver$.
	
	Thus, instead of 0, the probability of getting $c_\receiver$ upon a $\HIRE$ signal in the sender-optimal variant is $\frac{t}{n} \cdot \frac{n-t}{n-1} \cdot \frac{n-t-1}{n-2} \cdot \dots \cdot \frac{n-t-(x-2)}{n-(x-1)} \ge 0.$
	
	Further, a candidate that is not recommended for hire has an even smaller probability of being $c_\receiver$. The exact probability amounts to $\frac{t}{n}\cdot\frac{1}{t-1}\cdot \left(1-\frac{n-t}{n-1} \cdot \frac{n-t-1}{n-2} \cdot \dots \cdot \frac{n-t-(x-2)}{n-(x-1)}\right) \quad \le \quad \frac{t}{n(t-1)}$.
	
	The probabilities of getting $c_\receiver$ when hiring in a round during the sampling phase or in a round after the \HIRE\ signal remain $1/n$ by similar calculations as above.
	
	Overall, the incentive to follow the algorithm weakly increases for decreasing rank $x$. The mechanism is persuasive when $e_\receiver$ is the worst candidate for $\sender$, and also of $c_\receiver$ is at any rank $1 \le x < n$ in the preference of \sender. 
	
	Together with the fact that the mechanism clearly hires the sender-optimal candidate $c_\sender$ with probability $1/e-o(1)$ and thus yields an expected utility of $(1/e-o(1))\sv_{\max}$ for $\sender$, this proves the theorem.
\end{proof}

\subsection{Basic Scenario with Disclosure}
In this section, we consider the basic scenario with disclosure, i.e., both $\sender$ and $\receiver$ have the same a-priori information as in the basic scenario. Now the rejected candidates are revealed to $\receiver$. This additional information obtained by disclosure does not significantly change the results for $\sender$.

In round $t$, the Adaptive Elementary mechanism signals as follows. If there has been no $\HIRE$ signal in any earlier round and $c_\sender$ arrives in round $t$, it signals $\HIRE$. If there has been no $\HIRE$ signal in any earlier round and $c_\receiver$ arrives in round $t$, it signals $\HIRE$ with probability $1/(n-t)$. It signals $\NOHIRE$ in any other case.
\begin{lemma}
	The Adaptive Elementary mechanism is persuasive in the basic scenario with disclosure.
\end{lemma}
\begin{proof}
	If $c_\receiver$ has already arrived and was rejected, $\receiver$ will not get her best candidate. Recall that the rejected candidates are revealed. Since $\receiver$ is indifferent among the remaining candidates, it is optimal to follow the remaining signals of $\sender$. Hence, for the remainder of the persuasiveness argument, we assume that $c_{\receiver}$ has not been rejected.
	
	If the signal in round $t$ is $\HIRE$, the current candidate is $c_\receiver$ with probability 
	$\frac{1}{n-t} \Big/ \left(1+\frac{1}{n-t}\right) = \frac{1}{n-t+1}$. If $\receiver$ decides to deviate and hire in some later round $r$, she will not get any additional information. Thus, due to random-order arrival, the probability of hiring $c_{\receiver}$ in round $r$ is $\left(1 - \frac{1}{n-t+1}\right)  \cdot \frac{1}{n-t} = \frac{1}{n-t+1}$. Hence, it is optimal to follow the signal.
	
	Now suppose up to and including round $t$ there have been only $\NOHIRE$ signals. If $\receiver$ deviates and hires the current candidate in round $t$, she hires $c_\receiver$ with probability $\frac{n-t-1}{n-t}\Big/ \left(n-t-1+\frac{n-t-1}{n-t}\right) = \frac{1}{n-t} \Big/ \left(1+\frac{1}{n-t}\right) = \frac{1}{n-t+1}$. If instead $\receiver$ follows the mechanism, then suppose the $\HIRE$ signal comes in round $r > t$. We showed above that this gives a success probability of $\frac{1}{n-r+1}$ -- if $c_{\receiver}$ \emph{was not rejected in one of the rounds $i = t,\ldots,r-1$}. By the previous paragraph, in round $i$ with signal $\signal_i = \NOHIRE$, candidate $c_\receiver$ is getting rejected with probability $ \frac{1}{n-i+1}$. Thus, if $\receiver$ follows the mechanism, the success probability is $\prod_{i=t}^{r-1} \left(1-\frac{1}{n-i+1}\right) \cdot \frac{1}{n-r+1} = \frac{1}{n-t+1}$.
	
	Hence, it is optimal for $\receiver$ to follow the mechanism.
\end{proof}
\begin{proposition}\label{prop:receiverOrdinalKnownUtilDisclosure}
	In the basic scenario with disclosure, the Adaptive Elementary mechanism yields a success probability of $1-o(1)$ and an expected utility of $(1-o(1))\cdot \sv_{\max}$ for $\sender$.
\end{proposition}
\begin{proof}
	
	Conditional on the fact that the sender-optimal candidate $c_\sender$ has not shown up in rounds ${1,\dots,t-1}$, $c_\sender$ arrives in round $t$ with probability $1/(n-t+1)$. A similar probability holds for the receiver-optimal candidate $c_\receiver$ (i.e., conditional on $c_\receiver$ not having arrived in an earlier round, the current candidate being $c_\receiver$). 
	
	We first bound the success probability if $\receiver$ follows the mechanism. Let $A_j$ be the probability of hiring $c_\sender$ if $j \ge 2$ candidates remain and $c_\sender$ and $c_\receiver$ are both among those remaining candidates. The base case is $A_2 = 1/2$ as whichever candidate shows up first will be recommended with probability 1. We can express $A_j$ by the following recursion:
	\[A_j = \underbrace{\frac{1}{j}\cdot 1}_{c_\sender \text{ arrives}} + \underbrace{\frac{1}{j}\cdot \frac{j-2}{j-1}}_{c_{\receiver} \text{ arrives}} + \underbrace{\frac{j-2}{j}\cdot A_{j-1}}_{\text{neither}}\]
	Solving the recursion gives us
	\begin{align*}
	A_n \quad &= \quad \frac{1}{n} + \frac{1}{n}\cdot \frac{n-2}{n-1} + \frac{n-2}{n}\cdot A_{n-1} 
	\quad = \quad \frac{2n-3}{n(n-1)} + \frac{2n-5}{n(n-1)} + \frac{n-2}{n}\cdot \frac{n-3}{n-1}\cdot A_{n-2} \\
	&= \quad \frac{2n-3}{n(n-1)} + \frac{2n-5}{n(n-1)} + \frac{2n-7}{n(n-1)} + \frac{(n-4)(n-3)}{n(n-1)}\cdot A_{n-3} \\
	&= \quad \sum_{i=1}^{n-2} \frac{2(n-i)-1}{n(n-1)} + \frac{2}{n(n-1)} \cdot \frac{1}{2} 
	\quad = \quad 1-\frac{1}{n}\enspace.
	\end{align*}
	
	This proves the approximation guarantee. 
\end{proof}

\subsection{Secretary Scenario with Disclosure}

\begin{algorithm}[t]
	\caption{\label{algo:FirstBest} First-Opt Mechanism}
	\DontPrintSemicolon
	\KwIn{Number of candidates $n$, sample size $s > 1$}
	$A_0 \leftarrow \emptyset$ \;
	\For{$t=1$ to $n-1$}{
		$A_t \leftarrow A_{t-1} \cup \{\theta_t\}$\\ 
		\lIf{$t\le s$}{Signal $\NOHIRE$\hfill \texttt{// only $\NOHIRE$ in the sample phase}} 
		\ElseIf{$(\theta_t = \arg\max_{i \in A_t} \sv_i)$ or $(\theta_t = \arg\max_{i \in A_t} \rv_i)$}{Signal $\HIRE$ and end mechanism}
		\lElse{Signal $\NOHIRE$}}
	Signal $\HIRE$ on $n$-th candidate. \hfill \texttt{// by definition}
\end{algorithm}

In the secretary scenario with disclosure, $\receiver$ sees in each round $t$ a signal $\sigma_t$. When $\receiver$ decides to reject the candidate in round $t$, she gets to see its identity $\theta_t$. Recall that we assume $\receiver$ knows the set of candidates upfront, while $\sender$ only learns the values of a candidate when it arrives. We consider the First-Opt mechanism (see Algorithm~\ref{algo:FirstBest}). It rejects the first $s$ candidates. We will choose $s = \lfloor n/2 \rfloor$. Subsequently, it recommends for hire the first candidate that is best among the arrived candidates, either for $\sender$ or for $\receiver$. In the last round, $\sender$ always signals $\HIRE$.

\begin{lemma}
	The First-Opt mechanism is persuasive in the secretary scenario with disclosure.
\end{lemma}
\begin{proof}
	If $c_\receiver$ has already arrived and was rejected (and, thus, revealed), $\receiver$ is indifferent among the remaining candidates. As such, it is optimal to follow the remaining signals of $\sender$. For the remainder of the proof of persuasiveness, we assume $c_{\receiver}$ has not been rejected.
	
	Suppose $\sigma_t = \HIRE$ in round $t \le n-1$, then $\Pr[\theta_t = c_\receiver \mid \sigma_t = \HIRE] \ge \frac{1}{n-t+1}$. If $\receiver$ decides to deviate and hire in some later round $r > t$, she will not get any additional information from $\sender$. Due to random-order arrival, the success probability is $\Pr[\theta_r = c_{\receiver} \mid \sigma_t = \HIRE] \le \left(1 - \frac{1}{n-t+1}\right) \cdot \frac{1}{n-t} = \frac{1}{n-t+1}$. Hence, it is optimal to follow the signal.
	
	Now suppose $\sigma_{t'} = \NOHIRE$ for all $t' \in [t]$. If $t \ge s+1$ after the sample phase, the candidate is not the best one among the arrived ones, so $\theta_t \neq c_\receiver$. It is optimal for $\receiver$ to wait for a round with a $\HIRE$ signal. Otherwise, if $t \le s$ during the sample phase, then $\Pr[\theta_t = c_\receiver] = \frac{1}{n-t+1}$ (note that we assume $c_\receiver$ has not been rejected). In contrast, if $\receiver$ follows the mechanism, then $\Pr[\theta_t, \theta_{t+1},\ldots,\theta_s \neq c_\receiver] = \prod_{t'=t}^{s} \left(1 - \frac{1}{n-t'+1}\right) = \frac{n-s}{n-t+1}$. Then, the sample phase ends and, due to disclosure, $\receiver$ can determine all remaining candidates that are better (in terms of $\sv$, $\rv$, or both) than all arrived (and revealed) candidates during the sample phase. All remaining $n-s$ candidates might fulfill this condition, with $c_\receiver$ being one of them. If $\receiver$ follows the mechanism, she will end up hiring a random one of these candidates. Hence, the success probability when following the mechanism is at least $\frac{n-s}{n-t+1} \cdot \frac{1}{n-s} = \frac{1}{n-t+1}$.
	Hence, it is optimal for $\receiver$ to follow the mechanism.
\end{proof}
\begin{theorem}\label{theo:receiverOrdinalSecretaryDisclosure}
	In the secretary scenario with disclosure, the First-Opt mechanism with $s = \lfloor n/2 \rfloor$ yields a success probability of at least $1/4-o(1)$ and an expected utility of at least $(1/4-o(1))\cdot \sv_{\max}$ for $\sender$.
\end{theorem}

\begin{proof}
	Let $A_t$ denote the random set of candidates observed up to round $t$. The first $s$ candidates get a $\NOHIRE$ signal. For each round $t \ge s + 1$, we can determine the candidate in round $t$ by first drawing randomly the set $A_t$, and then drawing a uniform random candidate from $A_t$ to arrive in round $t$. In $A_t$ there are (one or) two candidates that can potentially generate a $\HIRE$ signal in round $t$, the best one for $\sender$ and the one best for $\receiver$ (could be the same one). Hence, the probability of signal $\HIRE$ is at most $\frac{2}{t}$. Overall, the probability $\prob[\signal_t = \NOHIRE \mid A_t]$ is at least
	\begin{equation}\label{eqn:inequalitySignalNot}
	\prob[\signal_t = \NOHIRE \mid A_t] \ge \begin{cases}
	1 & t = 1, \dots, s \\
	\frac{t-2}{t} & t = s+1, \dots, n-1
	\end{cases}.
	\end{equation}
	The success probability for $\sender$ in round $t > s$ conditioned on $A_t$ is
	\begin{align*}
	\Pr[\theta_t = c_\sender \mid A_t] &= \Pr[c_\sender \in A_t] \cdot \Pr[\theta_t = c_\sender \mid c_\sender \in A_t] \cdot \prob[\signal_1,\dots,\signal_{t-1} = \NOHIRE \mid A_{t-1}] \\
	&= \frac{t}{n} \cdot \frac{1}{t} \cdot \prod_{k=s+1}^{t-1}\prob[\signal_k = \NOHIRE \mid A_k] 
	\ge \frac{1}{n} \prod_{k=s+1}^{t-1} \frac{k-2}{k} \\
	&= \frac{1}{n}\frac{(s-1)s}{(t-2)(t-1)}\enspace,
	\end{align*}
	and, since these probabilities are independent of $A_t$, the overall success probability for $\sender$ is at least
	\begin{align*}
	\sum_{t=s+1}^{n} \frac{1}{n} \cdot \frac{(s-1)s}{(t-2)(t-1)} \quad = \quad \frac{(s-1)s}{n} \sum_{t= s+1}^n \left( \frac{1}{t-2} - \frac{1}{t-1} \right) \quad = \quad \frac{s}{n}\left(1 - \frac{s-1}{n-1}\right).
	\end{align*}
	The expression is optimized for $s= \lfloor n/2 \rfloor$ and yields a success probability of $\frac{1}{4} - o(1)$. The mechanism is entirely symmetric in terms of sender and receiver, so the same result applies to the success probability of $\receiver$.
\end{proof}

We now show that the First-Opt mechanism is indeed optimal, in the sense that there are cases in which no mechanism can achieve a better success probability. We study negatively correlated utilities, i.e., the candidate with rank $i$ for $\sender$ has rank $n-i+1$ for $\receiver$ for all $i \in [n]$. Observe that, beyond defining the ranking, cardinal utility values are irrelevant for the objectives of $\sender$ and $\receiver$. Moreover, given the ranking of known candidates, $\sender$ can infer no additional information from their values about the position of the current candidates in the overall ranking. As such, we will assume that $\sender$ ignores the cardinal values.

\begin{theorem}
	\label{thm:BeckmannUB}
	If utilities of sender and receiver are negatively correlated, the First-Opt mechanism maximizes the success probability for $\sender$ among all persuasive mechanisms in the secretary scenario with disclosure.
\end{theorem}

We first show the structural Lemmas~\ref{lem:BeckmannBestSoFar}-\ref{lem:Beckmann0-1}. We concentrate on the following class of randomized \emph{Best-So-Far} mechanisms. A Best-So-Far mechanism signals $\HIRE$ in rounds $1,\ldots,n-1$ with probability > 0 only if the candidate in the current round is best so far for either $\sender$ or $\receiver$. A Best-So-Far mechanism always signals $\HIRE$ in the last round if it has not done so before. We show that we can turn any mechanism into a Best-So-Far mechanism. This yields higher success probabilities for $\sender$ and does not hurt the incentives for $\receiver$, since hiring non-best-so-far candidates is unprofitable for both $\sender$ and $\receiver$.
\begin{lemma}
	\label{lem:BeckmannBestSoFar}
	If utilities of sender and receiver are negatively correlated, then for every persuasive mechanism there is a persuasive Best-So-Far mechanism with weakly higher success probability for $\sender$.
\end{lemma}
\begin{proof}
	Consider a persuasive mechanism $\phi$ for $\sender$. We construct a new Best-So-Far mechanism $\phi'$ as follows. The new mechanism runs $\phi$. Whenever $\phi$ decides to signal $\HIRE$ on a candidate that is neither best so far for $\sender$ nor $\receiver$, $\phi'$ changes the signaling strategy. We assume for the moment that in this case $\phi'$ sends to $\receiver$ a separate transition signal. $\phi'$ then signals $\NOHIRE$ in all rounds $t,\ldots,n-1$ and $\HIRE$ in round $n$. 
	
	Clearly, for any $\HIRE$ signal received in rounds $1,\ldots,n-1$, $\phi'$ yields a higher conditional probability that the recommended candidate is $c_\receiver$. If the transition signal is received, the mechanism switches to a deterministic $\HIRE$ signal in the last round. This only increases both, the probability to hire the sender- and receiver-optimal candidates $c_\sender$ and $c_\receiver$, since in this case $\phi$ would have hired a suboptimal candidate for both $\sender$ and $\receiver$.
	
	If $\phi$ comes to round $n-1$ without a $\HIRE$ signal, $\phi'$ signals $\HIRE$ in round $n$. This, too, weakly increases both the probabilities of hiring $c_\sender$ and of hiring $c_\receiver$.
	
	Hence, given that $\phi$ is persuasive, $\phi'$ is persuasive as well. It increases the success probability for both $\sender$ and $\receiver$. Finally, $\phi'$ remains persuasive even when we omit the transition signal, since $\receiver$ is given only less information.
\end{proof}

For the rest of the argument, we concentrate on Best-So-Far mechanisms. Note that Best-So-Far mechanisms are not always persuasive. For example, simply flipping a coin and running the sender-optimal or the receiver-optimal secretary algorithm throughout is a Best-So-Far mechanism. In the secretary scenario without revelation this was our persuasive approach, but here it might give $\receiver$ incentives to deviate: By the first time $\receiver$ sees that a currently best candidate for her was rejected, $\receiver$ knows the sender-optimal algorithm is running. Then she might have an incentive to deviate, since the sender-optimal algorithm never recommends $c_\receiver$.

We now enlarge the class of mechanisms under consideration. Suppose we are given any history $A_{t-1}$ of arrived candidates until round $t-1$. In round $t$, let $p_t^{\sender}$ be the probability of signaling $\sigma_t = \HIRE$ if $\theta_t$ is the best candidate so far for $\sender$. We define $p_t^{\receiver}$ accordingly. 
\begin{lemma}
	\label{lem:BeckmannSameP}
	If a Best-So-Far mechanism is persuasive for negatively correlated utilities, it satisfies $p_t^{\receiver} \ge p_t^{\sender}$ for all rounds $t \in [n]$ and all histories $A_{t-1}$. 
\end{lemma}
\begin{proof}
	Consider the beginning of round $t \le n-1$, where $n-t+1$ candidates are still to arrive. Since $\receiver$ knows all candidates upfront and sees all rejected ones, she can determine which candidates would qualify in round $t$ as best so far for $\sender$ or $\receiver$. Suppose $c_\receiver$ has not arrived yet. We construct a worst-case scenario for $\receiver$ as follows: The second-best candidate for $\receiver$ has arrived, so $c_\receiver$ is the only one that generates a $\HIRE$ signal in favor of $\receiver$. All other $n-t$ candidates are the top $n-t$ candidates for $\sender$.
	
	In this worst-case scenario, if the probability $p_{t}^{\sender} > p_t^{\receiver}$, then a $\HIRE$ signal in round $t$ leads to hiring of $c_\receiver$ with probability of $p_t^{\receiver} \Big/ (p_t^{\receiver} + (n-t)p_t^{\sender}) < 1/(n-t+1)$. In contrast, if a $\NOHIRE$ signal is sent and $\receiver$ deviates, she gets $c_\receiver$ with probability of $(1-p_t^{\receiver}) \Big/ (1 - p_t^{\receiver} + (n-t)(1-p_t^{\sender})) > 1/(n-t+1)$. Hence, the mechanism is clearly not persuasive.
	
	$\sender$ (unlike $\receiver$) is entirely unaware of whether the situation in the current round represents a worst-case scenario or not. In fact, in every round, every history of arrived candidates could give rise to such a worst-case scenario. As such, in order to guarantee persuasiveness, the mechanism needs to satisfy $p_t^{\receiver} \ge p_t^{\sender}$ for all rounds $t \in [n-1]$ and all histories of arrived candidates.
\end{proof}
Consider the class of Best-So-Far mechanisms with $p_t^{\receiver} \ge p_t^{\sender}$. It could potentially extend \emph{slightly beyond} persuasive ones. We now optimize the success probability for $\sender$ within this class. By the previous two lemmas, this gives an upper bound on the success probability of any persuasive mechanism. For larger values of $p_t^{\receiver}$, we obviously have smaller success probability for $\sender$. Thus, the best mechanism in the class satisfies $p_t = p_t^{\receiver} = p_t^{\sender}$ in each round and for each history. Moreover, we can show that $p_t \in \{0,1\}$.
\begin{lemma}
	\label{lem:Beckmann0-1}
	There is an optimal persuasive Best-So-Far mechanism for negatively correlated utilities such that $p_t \in \{0,1\}$ for all $t \in [n-1]$.
\end{lemma}
\begin{proof}
	The proof is a simple backwards induction. Given round $n$, we assume by definition that a Best-So-Far mechanism has $p_n = 1$. Consider round $n-1$ and condition on the event that the current candidate is best so far for either $\sender$ or $\receiver$.\footnote{Note that since $p_{n-1} = p_{n-1}^{\sender} = p_{n-1}^{\receiver}$, we may not distinguish between the cases best so far for $\sender$ or $\receiver$.} Now there are two options, hire in this round or hire in round $n$. $\sender$ can determine probabilities of hiring the sender-optimal candidate $c_\sender$ for both these options. It is then optimal for $\sender$ to decide deterministically for the option (i.e., either $p_{n-1} = 1$ or $p_{n-1} = 0$), whichever gives higher success probability.
	
	Thus, starting in round $n-1$, the optimal policy uses $p_n, p_{n-1} \in \{0,1\}$. In round $t < n-1$, we have a similar situation. If the candidate is best so far, there are two options - either hire in this round or reject and invoke the optimal policy for rounds $t+1,t+2,\ldots,n$. $\sender$ can determine probabilities of hiring $c_\sender$ for both these options. It is then optimal for $\sender$ to decide deterministically for the option (i.e., either $p_t = 1$ or $p_t = 0$), whichever gives higher success probability.
	
	Hence, by induction, the optimal Best-So-Far mechanism with $p_t = p_t^{\receiver} = p_t^{\sender}$ has all $p_t \in \{0,1\}$.
\end{proof}
\begin{proof}[Proof of Theorem~\ref{thm:BeckmannUB}]
	Based on Lemma~\ref{lem:Beckmann0-1}, our goal is to determine optimal values $p_t \in \{0,1\}$. Our proof generalizes ideas for the standard secretary problem~\cite{Beckmann90}. Let $B_t$ be the event that the candidate in round $t$ is best so far for either $\sender$ or $\receiver$. Let $R_t$ be the event that candidate $\theta_t$ is rejected. Let $R_{-t}$ be the event that all candidates $\theta_1,\ldots,\theta_{t-1}$ were rejected. Let $u_t = \Pr[c_\sender \text{ hired in rounds } t,\ldots,n \mid B_t, R_{-t}]$ and $v_t = \Pr[c_\sender \text{ hired in rounds } t+1,\ldots,n \mid R_t, R_{-t}]$.
	
	To reach round $t+1$, we must have $R_{-t}$ and either $R_t$, or $B_t$ and $p_t = 0$. In both cases, consider the conditional probability of $B_{t+1}$ i.e., $\Pr[B_{t+1} \mid R_t, R_{-t}]$ and $\Pr[B_{t+1} \mid B_t, R_{-t}]$.
	
	Note that due to random order arrival, conditioned on every set $A_t$ of candidates arrived in rounds $k = 1,\ldots,t$, we have the same probabilities of generating $B_k$ and $R_k$ events and, thus, the same probability of $B_t$, $R_{-t}$ and $R_t$. Now, conditioned on $A_{t+1}$, the probability to have $B_{t+1}$ in round $t+1$ is $\Pr[B_{t+1} \mid A_{t+1}] = \frac{2}{t+1}$, which is the same for every $A_{t+1}$. Since  $\Pr[R_{-t} \mid A_t]$ and $\Pr[R_t \mid A_t]$ are independent of the set $A_t$, we also have $\Pr[B_{t+1} \mid A_{t+1}, R_t, R_{-t}] = \Pr[B_{t+1} \mid R_t, R_{-t}] = \frac{2}{t+1}$. Moreover, since $\Pr[B_t \mid A_t]$ is the same for every $A_t$, we have $\Pr[B_{t+1} \mid A_{t+1}, B_t, R_{-t}] = \Pr[B_{t+1} \mid B_t, R_{-t}]=\frac{2}{t+1}$. Thus, we see that
	\begin{equation}\label{eqn:recursion-v}
	v_t =  \frac{2}{t+1} \cdot u_{t+1} + \frac{t-1}{t+1} \cdot v_{t+1} \enspace.
	\end{equation}
	Now suppose $\theta_t$ is best so far, i.e., condition on $B_t$. If we hire $\theta_t$, we get the sender-optimal $c_\sender$ with probability $\frac{t}{2n}$. Otherwise, if we reject it, we get $c_\sender$ with probability $v_t$. Thus, an optimal Best-So-Far mechanism will choose the better alternative and obtain
	\begin{equation}\label{eqn:recursion-u}
	u_t = \max \left\{\frac{t}{2n}, v_t\right\} \enspace.
	\end{equation}
	We resolve the recurrences for $u_t$ and $v_t$ by backwards induction. The base cases are $v_n = 0$ and $u_n = \frac{1}{2}$. Note that $v_{n-1} = \frac{n-2}{n} \cdot v_n + \frac{2}{n} \cdot u_n = \frac{1}{n}$, which yields $u_{n-1} = \max \left\{ \frac{n-1}{2n}, v_{n-1} \right\} = \frac{n-1}{2n}$. Repeating this argument, we have $u_t > v_t$ as long as $\frac{t}{2n} > v_t$, and thus should set $p_t = 1$. If on the other hand $v_t \ge \frac{t}{2n}$, it is optimal to wait and set $p_t = 0$, even though the current candidate is best so far. More generally, $u_t$ and $v_t$ can be given as follows which we prove below.
	
	\begin{lemma}
		\label{lem:BeckmannRecursion}
		$u_t$ and $v_t$ are given by
		\begin{equation*}
		u_t = \begin{cases}
		\frac{t}{2n} & t \ge \frac{n+1}{2} \\
		\frac{n}{4(n-1)} & t < \frac{n+1}{2}, n \text{ even,} \\
		\frac{n+1}{4n} & t < \frac{n+1}{2}, n \text{ odd,}
		\end{cases}
		\hspace{2cm}
		v_t = \begin{cases}
		\frac{t(n-t)}{n(n-1)} & t \ge \frac{n+1}{2} \\
		\frac{n}{4(n-1)} & t < \frac{n+1}{2}, n \text{ even,} \\
		\frac{n+1}{4n} & t < \frac{n+1}{2}, n \text{ odd.}
		\end{cases}
		\end{equation*}
	\end{lemma}
	Hence, it is optimal to signal $\NOHIRE$ in the first $s = \lfloor n/2 \rfloor$ rounds and then signal $\HIRE$ for the first candidate that is best so far, for either $\receiver$ or $\sender$. This is the First-Opt mechanism, and Theorem~\ref{thm:BeckmannUB} is proved.
\end{proof}
\begin{proof}[Proof of Lemma \ref{lem:BeckmannRecursion}]
	First, consider the case $t \ge \frac{n+1}{2}$ and start with $t = n$. Due to the base cases, $v_n = 0 = \frac{n(n-n)}{n(n-1)}$ and $u_n = \frac{1}{2} = \max \left\{\frac{n}{2n}, 0 \right\}$.
	
	Assume the lemma holds for $t+1 \ge \frac{n+1}{2}+1$. We show that it holds for $t$ as well. First, observe that $\frac{t}{2n} \ge \frac{t(n-t)}{n(n-1)}$ for $t \ge \frac{n+1}{2}$ and thus $u_t = \frac{t}{2n}$. For $v_t$ we get the following:
	\begin{align*}
	v_t &= \frac{t-1}{t+1} \cdot v_{t+1} + \frac{2}{t+1} \cdot u_{t+1} \\
	&= \frac{t-1}{t+1} \cdot \frac{1}{n} \cdot \frac{(t+1)(n-(t+1))}{n-1} + \frac{2}{t+1} \frac{t+1}{2n} \\
	&= \frac{1}{n} \cdot \frac{(t-1)(n-t-1) + (n-1)}{n-1} \\
	&= \frac{1}{n} \cdot \frac{t(n-t)}{n-1}
	\end{align*}
	Now consider the second case in which $t < \frac{n+1}{2}$. The base case is $t = \frac{n}{2}$ or $t = \frac{n-1}{2}$, depending on the parity of $n$. 
	
	\begin{description}
		\item[Case $n$ even:] It holds
		\begin{align*}
		v_{\frac{n}{2}} &= \frac{n/2-1}{n/2+1}\cdot v_{\frac{n}{2}+1} + \frac{2}{n/2+1}\cdot u_{\frac{n}{2}+1} \\
		&= \frac{n-2}{n+2} \cdot \frac{1}{n} \cdot \frac{(n+2)(n-2)}{4(n-1)} + \frac{2}{n/2+1} \cdot \frac{n/2+1}{2n} \\
		&= \frac{1}{n} \cdot \frac{(n-2)^2+4(n-1)}{4(n-1)} \\
		&= \frac{1}{4n(n-1)} \left( n^2 - 4n + 4 + 4n -4 \right) \\
		&= \frac{n}{4(n-1)}
		\end{align*}
		Since $t = \frac{n}{2} < \frac{n+1}{2}$ and thus $\frac{\frac{n}{2}}{2n} = \frac{1}{4} < \frac{n}{4(n-1)} = v_t$,  we now have $u_t = v_t$.
		This obviously leads to $v_t = v_{t+1}$ and thus $u_t = v_t$ for all $t < \frac{n+1}{2}$.
		
		\item[Case $n$ odd:] The case is similar:
		\begin{align*}
		v_{\frac{n-1}{2}} &= \frac{(n-1)/2-1}{(n+1)/2}\cdot v_{\frac{n+1}{2}} + \frac{2}{(n+1)/2}\cdot u_{\frac{n+1}{2}} \\
		&= \frac{n-3}{n+1} \cdot \frac{1}{n} \cdot \frac{(n+1)/2((n-1)/2)}{n-1} + \frac{2}{(n+1)/2}\cdot \frac{(n+1)/2}{2n} \\
		&= \frac{1}{n} \cdot \frac{1}{4(n-1)} \left((n-3) (n-1) + 4(n-1) \right) \\
		&= \frac{1}{n} \cdot \frac{1}{4(n-1)} \left(n^2 - 4n + 3 + 4n-4 \right) \\
		&= \frac{1}{n} \cdot \frac{1}{4(n-1)} (n+1)(n-1) \\
		&= \frac{n+1}{4n}
		\end{align*}
		Again, $t = \frac{n-1}{2} < \frac{n+1}{2}$ and thus $\frac{\frac{n-1}{2}}{2n} = \frac{1}{4} - \frac{1}{4n} < \frac{n+1}{4n} = v_t$. Hence, $u_t = v_t$ for all $t < \frac{n+1}{2}$.
	\end{description}
	This proves the lemma.
\end{proof}

\end{document}